\documentclass[journal ,12]{IEEEtran}
\usepackage[pdftex]{graphicx}
\graphicspath{{../pdf/}{../jpeg/}}
\DeclareGraphicsExtensions{.pdf,.jpeg,.png}
\usepackage[cmex10]{amsmath}
\usepackage{amsthm}
\usepackage{graphicx,color}
\usepackage[dvipsnames]{xcolor}
\usepackage{graphicx}
\usepackage{mathtools}
\usepackage{amssymb,color}
\allowdisplaybreaks
\usepackage{subcaption}
\usepackage{cite}
\usepackage{blindtext}
\usepackage{tcolorbox}
\usepackage{colortbl}

\definecolor{Gray}{gray}{0.9}
\usepackage{xcolor}
\usepackage{lipsum}
\usepackage[font={small}]{caption}
\usepackage{algorithm}
\usepackage{algorithmic}
\usepackage{array} 
\usepackage{makecell}
\usepackage{tabularx}
\usepackage{bm}
\usepackage{siunitx, soul}
\usepackage{cases}

\usepackage{pgfplots}
\pgfplotsset{compat=newest}
\usetikzlibrary{plotmarks}
\usetikzlibrary{arrows.meta}
\usepgfplotslibrary {patchplots}
\usepackage{grffile}
\usetikzlibrary{plotmarks}

\usepackage{booktabs,caption}
\usepackage[flushleft]{threeparttable}

\usepackage{xcolor}
\usepackage{soul}

\newtheorem{theorem}{Theorem}

\newtheorem{rem}{Remark}
\usepackage{booktabs} 

\usepackage[font=scriptsize,labelfont=bf]{caption}
\usepackage{pifont}
\newcommand{\cmark}{\ding{51}}%
\newcommand{\xmark}{\ding{55}}%
\newcommand{\lmark}{\textbf{--}}

\newcommand{\tr}{\text{Tr}}

\newcommand{\bc}{\text{BackCom}\xspace}

\newcommand{\abc}{\text{AmBC}\xspace} 
\newcommand{\thr}{\text{th}}

\makeatother
\bstctlcite{IEEEexample:BSTcontrol}
\newcommand{\qh}{\mathbf{h}}

\newcommand{\qf}{\mathbf{f}}
\newcommand{\qs}{\mathbf{s}}
\newcommand{\qg}{\mathbf{g}}
\newcommand{\qr}{\mathbf{r}}
\newcommand{\qu}{\mathbf{u}}

\newcommand{\qw}{\mathbf{w}}
\newcommand{\qx}{\mathbf{x}}

\newcommand{\qS}{\mathbf{S}}

\newcommand{\qG}{\mathbf{G}}
\newcommand{\qW}{\mathbf{W}}
\newcommand{\qQ}{\mathbf{Q}}
\newcommand{\qI}{\mathbf{I}}
\newcommand{\qR}{\mathbf{R}}
\newcommand{\qP}{\mathbf{P}}
\newcommand{\qalpha}{\boldsymbol{\alpha}}

\newcommand{\qU}{\mathbf{U}}

\definecolor{LightBlue}{rgb}{0.88,0.95,1}
\definecolor{DarkBlue}{rgb}{0.0,0.2,0.6}
\usepackage{enumitem}
\usepackage{multirow}

\makeatother
\bstctlcite{IEEEexample:BSTcontrol}
\definecolor{LightGray}{gray}{0.9}
\definecolor{MediumGray}{gray}{0.5}
\definecolor{DarkGray}{gray}{0.2}

\begin{document}
\bstctlcite{IEEEexample:BSTcontrol}

\title{Optimization of Rate-Splitting Multiple Access with Integrated Sensing and Backscatter Communication\vspace{-0mm}}

\author{Diluka Galappaththige, \IEEEmembership{Member, IEEE}, Shayan Zargari,   Chintha Tellambura, \IEEEmembership{Fellow, IEEE}, and Geoffrey Ye Li, \IEEEmembership{Fellow, IEEE,}
\thanks{D. Galappaththige, S. Zargari, and C. Tellambura with the Department of Electrical and Computer Engineering, University of Alberta, Edmonton, AB, T6G 1H9, Canada (e-mail: \{diluka.lg, zargari, ct4\}@ualberta.ca). \\
\indent G. Y. Li is with the ITP Lab, the Department of Electrical and Electronic Engineering, Imperial College London, SW7 2BX London, U.K.(e-mail: geoffrey.li@imperial.ac.uk).} \vspace{-0mm} }

\maketitle

\begin{abstract} 
An integrated sensing and backscatter communication (ISABC) system is introduced herein. This system features a full-duplex (FD) base station (BS) that seamlessly merges sensing with backscatter communication and supports multiple users. Multiple access (MA) for the user is provided by employing rate-splitting multiple access (RSMA).  RSMA, unlike other classical orthogonal and non-orthogonal MA schemes, splits messages into common and private streams.  With RSMA, the set of common rate forms can be optimized to reduce interference. Optimized formulas are thus derived for communication rates for users, tags, and the BS's sensing rate, with the primary goal of enhancing the transmission efficiency of the BS. The optimization task involves minimizing the BS's overall transmission power by jointly optimizing the BS's beamforming vectors, the tag reflection coefficients, and user common rates. The alternating optimization method is employed to address this challenge. Concrete solutions are provided for the received beamformers, and semi-definite relaxation and slack-optimization techniques are adopted for transmit beamformers and reflection coefficients, respectively. For example,  the proposed RSMA-assisted ISABC system achieves a  \qty{84.5}{\percent} communication rate boost over a non-orthogonal multiple access-assisted ISABC, with only a \qty{24}{\percent} increase in transmit power, leveraging ten transmit/reception antennas at the BS.
\end{abstract}


\IEEEpeerreviewmaketitle
\section{Introduction}
Integrated sensing and communications (ISAC) networks are envisioned for sixth-generation (6G) wireless. In traditional radar-based sensing, a reflected signal from an object/target is used to sense/detect the object \cite{3GPPISAC2024}. Similarly, backscatter communication (\bc) uses radar-like load modulation to send data,  i.e., a backscatter device (or a tag) modulates its data onto an external radio frequency (RF)  signal and reflects it to a reader \cite{HoangBook2020, Diluka2022, Rezaei2023Coding, Rezaei2020}.   In \cite{Diluka2023}, by leveraging the similarities between \bc and radar systems, i.e., the use of reflected signals, a novel concept called \textit{Integrated Sensing and Backscatter Communications (ISABC)}, has been proposed. This work investigates the feasibility of integrating sensing into low-power Internet-of-Things (IoT) devices. This new paradigm combines integrated sensing and communication (ISAC) with \bc features, resulting in concurrent sensing and communication in ambient power-enabled IoT networks.

\begin{table}[t]
    \centering
    \renewcommand{\arraystretch}{1.1} 
    \setlength{\tabcolsep}{10pt} 
    \caption{A comparison between ISAC and ISABC.}
    \begin{tabular}{|l|c|c|}
    \hline
    \textbf{Features} & \textbf{ISAC}  &  \textbf{ISABC} \\ \hline \hline
    Target & \checkmark & $\times$ \\\hline
    Tag & $\times$ & \checkmark \\ \hline
    Additional user data  & $\times$ & \checkmark \\\hline
    BS Power allocation  & \checkmark & \checkmark \\\hline
    User decoding & Conventional & SIC \\\hline
    Sensing type & Active/Passive & Active \\ 
    \hline
    \end{tabular}
    \label{tab:new_comparison}
    \vspace{-0mm}
\end{table}

Combining the strengths of ISAC and \bc, ISABC introduces backscatter tags as replacements for the conventional sensing targets, enabling opportunistic sensing in \bc systems \cite{Diluka2023}. Although ISABC can be considered a device-based ISAC variant with active sensing, this replacement distinguishes ISABC from traditional ISAC in several ways (Table \ref{tab:new_comparison}). In  ISAC, sensing targets can be passive objects, such as vehicles, which neither transmit nor receive sensing signals, nor active devices, such as mobile phones, which engage in transmitting and/or receiving.  In contrast, in ISABC, backscatter tags act as sensing targets. These tags relay environmental details to the base station (BS) and transmit additional data to users or readers. Thus, the BS can use the tag-reflected signals for sensing. This approach enhances communication and sensing capabilities by harnessing both sensing and backscatter data. However, the user/reader may need more advanced decoding mechanisms, such as successive interference cancellation (SIC).

6G IoT networks require accurate sensing, high data throughput, low latency, connectivity, reliability, and energy efficiency (EE) \cite{Huawei_ambient, Huawei}. Ambient power-enabled (battery-free) IoT applications, such as smart homes, smart cities, industrial IoT, and environmental monitoring, rely on sensing and exchanging environmental data without explicit human input \cite{HoangBook2020, Diluka2022, Rezaei2023Coding}. For instance, in smart homes, tags can function as sensors, reflecting signals to the BS for parameter estimation and monitoring while modulating and transmitting data for communication.

The energy-efficient dual sensing and communication capability of ISABC makes it ideal for applications mentioned above \cite{Huawei_ambient, Huawei}. In dense environments like industrial IoT or smart cities,  interference reduction improves resource allocation, enhancing network performance. Furthermore, backscatter-aided vehicular networks support navigation, traffic and pedestrian monitoring, environmental surveillance, road safety, and autonomous driving, requiring robust sensing capabilities \cite{Khan2021, Xu2023}. ISABC effectively addresses these demands, enabling such applications.

Interference issues in dense IoT environments, energy constraints, and hardware limitations must be addressed for successful deployment. The ISABC system should be tailored to work with IoT devices with limited processing power and rely on constrained energy sources or energy harvesting (EH). Aligning with 6G standards and technologies, particularly the integration of sensing, communication, and localization, further enhances the ISABC's applicability {\cite{Huawei_ambient, Huawei}}. Future field testing and simulations could validate ISABC's performance in realistic scenarios, confirming its potential for adoption in 6G IoT networks. 

Conversely, ISABC can build upon existing {\bc} studies {\cite{Liu2023Covert, Caihui2023}}, further optimizing resource allocation and enhancing performance. For instance, multi-antenna tags in {\cite{Liu2023Covert}} can significantly improve both sensing accuracy and communication efficiency. In contrast, integrating commercial off-the-shelf hardware in {\cite{Caihui2023}} into our framework could facilitate a cost-effective implementation. Through these advancements, our work may bridge the gap between high-performance sensing and efficient {\bc}, making it a promising approach for future IoT applications.

This research extends the notion of \cite{Diluka2023, Zargari2024} to a general framework with enhanced spectrum efficiency by integrating it with rate-splitting multiple access (RSMA) (see Section~\ref{sec_motiv}). 

\subsection{ISAC, RSMA, and BackCom}

Since our study integrates these three technologies, this section briefly describes them. 

ISAC reshapes 6G networks by enabling simultaneous communication and sensing through a shared hardware platform and signal processing framework \cite{Liu2022ISAC, Azar2024, Diluka2024NF}. This dual functionality allows the ISAC base station to communicate with users while sensing environmental information to support services such as localization, activity monitoring, object detection, and mapping. Furthermore, ISAC enhances communication by leveraging sensed data for precise beamforming, rapid failure recovery, and reduced channel estimation overhead \cite{Liu2022ISAC, Azar2024, Diluka2024NF}.

RSMA is an advanced multiple-access scheme that effectively controls interference by splitting user messages into common and private parts \cite{Clerckx2016, Mao2022, Xu2021}. The common parts are combined into a single stream, while private parts are encoded separately. Each user receiver first decodes and removes the common stream before decoding the private stream, using successive interference cancellation (SIC). This method offers superior spectral and energy efficiency compared to conventional schemes like space-division multiple access (SDMA), which treats interference as noise, and non-orthogonal multiple access (NOMA), which fully decodes interference \cite{Clerckx2016, Mao2022, Xu2021}.

\bc facilitates energy-efficient IoT networks \cite{HoangBook2020, Diluka2022, Rezaei2023Coding}. It enables wireless nodes (tags or IoT devices) to communicate by reflecting external radio frequency (RF) signals.  Thus, tags can operate without active RF components. This approach may be a sustainable alternative for battery-dependent IoT devices, which incur high maintenance costs, environmental impact, and safety issues \cite{HoangBook2020, Diluka2022, Rezaei2023Coding}. These tags, without active RF components and reflecting external RF, achieve ultra-low cost and energy use, ranging from nanowatts to microwatts \cite{HoangBook2020}.

\subsection{Motivation and Our Contribution}\label{sec_motiv}
This study introduces a novel and generalized system model for ISABC: multiple tags, multiple users, a backscatter reader, and a full-duplex (FD) BS (Fig.~\ref{fig_SystemModel}). This network combines primary users, tags, a reader,  and sensing into a unified framework, posing three critical challenges.
\begin{enumerate}
    \item The first challenge is providing multiple access for primary users.   Conventional solutions are time-division multiple access (TDMA), NOMA, and SDMA. In  SDMA, the rates can plateau even with excess transmit powers if the number of transmit antennas is insufficient or a resource block has multiple users \cite{Mao2018Rate}.  In contrast, NOMA handles interference by employing superposition coding and SIC. However, NOMA  has stringent decoding and SIC requirements, and channel gains must have a significant disparity \cite{Mao2018Rate}.  RSMA outperforms these regarding spectrum and energy efficiencies \cite{Mao2018Rate}. For these reasons,  RSMA-assisted ISABC is proposed for this study and evaluated against NOMA and SDMA benchmarks. 
    \item The second challenge is for the reader to detect the tag signals significantly weaker than direct signals due to the double path loss effect \cite{HoangBook2020, Diluka2022, Rezaei2023Coding}. Thus, the direct link interference (DLI) must be compensated for \cite{Long2020}. This assumes that the reader first decodes the users' data and subtracts the user interference or the DLI before decoding the tags data. Although this can be realized if the reader and BS exchange information through a control link \cite{Long2020, positioningLTE}, DLI cannot be entirely eliminated due to hardware limitations such as the reader's limited dynamic range. Thus, imperfect SIC to partially remove DLI occurs before decoding the backscatter data \cite{Long2020}.
    
    \item The third challenge is the  BS's ability to sense the tag signals, while strong self-interference (SI) affects its sensing  \cite{Zhenyao2023}. To successfully sense the tag-reflected signal, the BS must reduce SI power to a level comparable to the tags' backscattered signal power \cite{Zhenyao2023}. The SI suppression is not perfect.  Thus, the BS will operate with a residual SI. 
\end{enumerate}

Prior works \cite{Diluka2023} and \cite{Zargari2024} have several limitations.    Specifically,  \cite{Diluka2023}  considers only one backscatter tag and a user, a special case of our model (Fig. 1). Thus, \cite{Diluka2023} does not consider multiple access for users.     Reference \cite{Zargari2024} extends  \cite{Diluka2023} for multiple tags and a single user. However, it uses perfect SIC for backscatter data decoding at the user and perfect SI cancellation at the BS for acquiring sensing information. Conversely, this work generalizes the system model and addresses the technical limitations in   \cite{Diluka2023, Zargari2024} by investigating multi-tag with EH, RSMA-enabled multiple access for the primary users, and effects of imperfect SIC and SI cancellation. Thus, the optimization framework differs significantly from \cite{Zargari2024}.

{On the other hand, references {\cite{Du2024ConcurScatter, Du2024Orthcatter, Bowen2023}} address these challenges from various perspectives. Specifically, {\cite{Du2024ConcurScatter}} introduces an OFDM-enabled backscatter system that uses subcarrier pattern diversity to support multiple tags simultaneously. In {\cite{Du2024Orthcatter}}, an in-band OFDM backscatter system employs over-the-air code division to cancel co-channel interference. {\cite{Bowen2023}} proposes a constructive interference scheme that leverages direct link interference to enhance the backscattered signal and maximize the received SNR at the reader. While these approaches address some critical challenges, they overlook the impact of backscatter transmissions on primary system performance and do not utilize backscattered signals to extract environmental information from tags.}

Inspired by the challenges above and to explore the potential applications of ISABC, we present a novel RSMA-enabled ISABC system. Our  contributions are summarized as follows:
\begin{enumerate}
    \item  The following system model is developed. The BS uses RSMA to serve primary users and enhance their communication performance. The tags reflect the BS's signal to communicate with the reader, while the BS uses the same reflected signal to infer the tags' environmental information. 

    \item The  BS transmit power minimization problem, $\mathbf{P}_1$ \eqref{P1_prob},  is formulated. It preserves all nodes' quality-of-service (QoS) requirements, including tags and users. The optimization variables are the BS transmit/received beamforming, tag reflection coefficients, and users' common rates. Due to the products of the optimization variables,  $\mathbf{P}_1$ is non-convex, and the widely available convex algorithms are not applicable. 
    
    \item To address this challenge, an alternative optimization (AO) strategy \cite{bezdek2003convergence} is employed, beginning by optimizing the BS received beamforming for the tags' signal using minimum mean-squared error (MMSE) filtering and the generalized Rayleigh quotient form of the signal-to-noise-to-interference ratio (SINR). The semidefinite relaxation (SDR) approach is then used to determine the BS transmit beamforming and the user common rates  \cite{so2007approximating, Qingqing2019}. Finally, the feasibility problem of tag reflection coefficient optimization is transformed into a slack-optimization problem to obtain an effective solution \cite{Shayan2021, Razaviyayn2013}.

    \item The benefits of the proposed RSMA-assisted ISABC are compared to NOMA-assisted ISABC, RSMA-assisted \bc, conventional \bc, conventional ISABC, and sensing-only schemes (with EH). The convergence and complexity are also analyzed. With a configuration of ten BS antennas each for transmission and reception, the proposed scheme offers a \qty{84.5}{\percent}  communication rate enhancement compared to the NOMA-assisted ISABC while only requiring \qty{24}{\percent}  transmit power rise.    
\end{enumerate}

\subsection{Previous RSAM Contributions on ISAC and \bc}
{While  RSMA-assisted ISABC has not been studied before, several works have explored RSMA-assisted ISAC {\cite{Yin2022ISAC, Liu2024, He2024, Chen2024}} and RSMA-assisted {\bc} {\cite{Vu2024}}. In particular, reference {\cite{Yin2022ISAC}} studies an RSMA-assisted ISAC system with multiple users and a single target. The BS beamforming is designed to jointly minimize the Cram\'{e}r-Rao bound (CRB) of target estimation and maximize the minimum rate among users, demonstrating the superiority of RSMA compared to SDMA and NOMA in ISAC networks. Reference {\cite{Liu2024}} employs RSMA  as the multiple access technique in a multi-user, single-target ISAC satellite system. Using an iterative algorithm based on successive convex approximation (SCA) and Dinkelbach's method, beamforming at the satellite transmitter is designed to maximize the system's EE. Reference {\cite{He2024}} investigates a multi-vehicle collaborative sensing scheme in an RSMA-assisted ISAC system with a single target. It derives the sum rate lower bound under imperfect channel state information (CSI) and the CRB. Additionally, the CRB is minimized by optimizing the rate splitting ratio, BS beamforming, and transmit power using an SDR-based algorithm. In {\cite{Chen2024}}, BS beamforming is designed to jointly maximize the minimum user rate and minimize the largest eigenvalue of the CRB matrix in a multi-user, multi-target ISAC system. An SDR-based algorithm is proposed, demonstrating RSMA's superiority over SDMA in balancing the communication and sensing trade-off.}

Reference {\cite{Vu2024}} presents an RSMA-assisted symbiotic {\bc} system with two primary users and a single tag. The closed-form outage probability is obtained by designing beamforming weights with four different gain-control techniques.

Table~{\ref{tab_comparison}} provides a summary of relevant works, comparing this work with other RSMA-assisted ISAC and {\bc} literature. Substantial differences in system design, problem formulation, and algorithmic approach are evident. Importantly, these previous works do not consider ISABC. Hence, our key technological contribution is the beamforming design for an RSMA-assisted multi-user, multi-tag ISABC system, addressing the sensing requirements in future low-power IoT networks.

\begin{table*}[!t]
\centering
\begin{threeparttable}
\renewcommand{\arraystretch}{1.2}
\caption{{A comparison of RSMA-related literature.}}
\label{tab_comparison}
\begin{tabular}{|c| p{2.0cm}|lllc|p{5.5cm}|p{1.5cm}|}
\hline 
\multicolumn{1}{|c|}{} & \multicolumn{1}{c|}{} & \multicolumn{4}{c|}{System}                                                             & \multicolumn{1}{c|}{} & \multicolumn{1}{c|}{} \\ \cline{3-6}
\multicolumn{1}{|c|}{\multirow{-2}{*}{Ref.}}                      & \multicolumn{1}{c|}{\multirow{-2}{*}{Configuration}}                               & \multicolumn{1}{c|}{FD BS} & \multicolumn{1}{c|}{User} & \multicolumn{1}{c|}{Target} & Tag & \multicolumn{1}{c|}{\multirow{-2}{*}{Objective}}                           & \multicolumn{1}{c|}{\multirow{-2}{*}{Algorithm}}                           \\ \hline \hline
\cite{Yin2022ISAC}   &  & \multicolumn{1}{c|}{\xmark}   & \multicolumn{1}{c|}{Multiple}   & \multicolumn{1}{c|}{Single}  &  \lmark   & Jointly minimize CRB and maximize $\mathcal{R}_{\min}$  &  SCA \\ \cline{1-1} \cline{3-8}
\cite{Liu2024}   &  & \multicolumn{1}{c|}{\xmark}   & \multicolumn{1}{c|}{Multiple}   & \multicolumn{1}{c|}{Single}  &  \lmark   & Maximize the max-min fairness EE with low-resolution DACs  &  SCA, Dinkelbach \\ \cline{1-1} \cline{3-8}
\cite{He2024}   &  & \multicolumn{1}{c|}{\xmark}   & \multicolumn{1}{c|}{Multiple}   & \multicolumn{1}{c|}{Single}  & \lmark   & Minimize CRB subjected to $\mathcal{R}_u$  &  SDR \\ \cline{1-1} \cline{3-8}  
\cite{Chen2024}   & \multirow{-4}{*}{RSMA+ISAC} & \multicolumn{1}{c|}{\xmark}   & \multicolumn{1}{c|}{Multiple}   & \multicolumn{1}{c|}{Multiple}  &  \lmark  & Jointly maximize $\mathcal{R}_{\min}$ and minimize the largest
eigenvalue of the CRB matrix  &  SDR \\ \hline
\cite{Vu2024}   & RSMA+\bc & \multicolumn{1}{c|}{\xmark}   & \multicolumn{1}{c|}{Two}   & \multicolumn{1}{c|}{\lmark}  &   Single  & Closed-form expression for outage probability  & \lmark  \\ \hline
\textbf{This paper}   & RSMA+ISABC & \multicolumn{1}{c|}{\cmark}   & \multicolumn{1}{c|}{Multiple}   & \multicolumn{1}{c|}{\lmark}  & Multiple    & Minimize $P_t$ subject to $\mathcal{R}_u$, $\mathcal{R}_b$, tags' EH,  and $\mathcal{R}_s$  & AO, SDR, SCA  \\ \hline
\end{tabular}
\begin{tablenotes}
      \scriptsize{
      \item $\mathcal{R}_{\min}$ - Minimum user rate. \quad  $\mathcal{R}_{u}$ - User rate. \quad $\mathcal{R}_b$ - Tag backscatter rate. \quad $\mathcal{R}_s$ - Sensing rate. \quad $P_t$ - BS transmit power. \quad  DAC - Digital-to-analog converter. }
    \end{tablenotes}
  \end{threeparttable}\vspace{-0mm}
\end{table*}

\textit{Notation}: 
Boldface lower case and upper case letters represent vectors and matrices, respectively. For matrix $\mathbf{A}$, $\mathbf{A}^{\rm{H}}$ and $\mathbf{A}^{\rm{T}}$ are Hermitian conjugate transpose and transpose of matrix $\mathbf{A}$, respectively. $\mathbf{I}_M$ denotes the $M$-by-$M$ identity matrix. The Euclidean norm,   absolute value, and expectation operators are denoted by $\|\cdot\|$ and $|\cdot|$, and $\mathbb{E}\{\cdot\}$, respectively. A circularly symmetric complex Gaussian (CSCG) random vector with mean $\boldsymbol{\mu}$ and covariance matrix $\mathbf{C}$ is denoted by $\sim \mathcal{C}\mathcal{N}(\boldsymbol{\mu},\mathbf{C})$. Besides, $\mathbb{C}^{M\times N}$ and ${\mathbb{R}^{M \times 1}}$ represent $M\times N$ dimensional complex matrices and $M\times 1$ dimensional real vectors, respectively. Further, $\mathcal{O}$ expresses the big-O notation. Finally, $\mathcal{L} \triangleq \{1,\ldots,L\}$,  $\mathcal{K} \triangleq \{1,\ldots,K\}$,  $\mathcal{L}_l \triangleq \mathcal{L}\setminus\{l\}$, $\mathcal{K}_k \triangleq \mathcal{K}\setminus\{k\}$, $\mathcal{L}_l' \triangleq \{1,\ldots,l-1\}$, and $\mathcal{L}_l'' \triangleq \{l+1,\ldots,L\}$.

\section{System Components}\label{Sec_system_modelA}
\subsection{System Model}
An RSMA-assisted ISABC system (Fig.~\ref{fig_SystemModel}) is considered. It consists of an FD BS as the primary transmitter having $M$ transmit and $N$ receiver uniform linear array (ULA) antennas, $L$ single-antenna users ($U_l$ denotes the $l$-th user), $K$ single-antenna tags ($T_k$ denotes the $k$-th tag), and a single-antenna reader. The BS antennas are spaced at half-wavelengths \cite{Zhenyao2023}. Tags use power-splitting (PS) protocol for EH and data backscattering (Section~\ref{sec_EH}).

The BS communicates and senses the surroundings using transmit beamforming, whereas tags use the BS signal for EH and data backscattering via the PS protocol. Cooperation is assumed between the BS and the backscatter reader via a control link \cite{positioningLTE}. Thus, the reader subtracts the direct signal, i.e., $\qf_0\qx$, from its received signal before decoding the tag data. This  SIC operation minimizes the interference from the BS.
The BS also extracts environmental insights from unintentionally received backscatter signals \cite{Diluka2023}. To reduce self-interference (SI), the BS uses two distinct antenna sets for transmission and reception \cite{Zhenyao2023}.  Time synchronization of the nodes is assumed for simplicity \cite{Long2020}.

\subsection{Channel Model}
This study employs block flat-fading channel models. During each fading block, $\mathbf{f}_0 \in \mathbb{C}^{M\times 1}$, $\mathbf{f}_l \in \mathbb{C}^{M\times 1}$, and $\mathbf{g}_{f,k} \in \mathbb{C}^{M\times 1}$, respectively, represent the channels between the BS and the reader, the BS and $U_l$, and the BS and $T_k$. The channels between $T_k$ and $U_l$, and $T_k$ and the reader are denoted by $v_{l,k}$ and $q_k$, respectively, while $\mathbf{g}_{b,k} \in \mathbb{C}^{N\times 1}$ represents the channel between $T_k$ and the BS receiver antennas. Within these channels, pure communication channels, i.e., $\mathbf{f}_0$, $\mathbf{f}_l$, $v_{l,k}$ and $q_k$, are modeled as Rayleigh fading and given as $\mathbf{a} = \zeta_a^{1/2} \tilde{\mathbf{a}}$, where $\mathbf{a} \in \{\mathbf{f}_0, \mathbf{f}_l, v_{l,k}, q_k\}$ and $\zeta_a$ is the large-scale path-loss and shadowing,  which stays constant for several coherence intervals. Moreover, $\tilde{\mathbf{a}} \sim \mathcal{CN}(\mathbf{0}, \mathbf{I}_{A})$ accounts for the small-scale Rayleigh fading, where $A\in \{M,1\}$.

Conversely, following the echo signal representation in multiple-input and multiple-output (MIMO) radar systems, the channels between the BS and tags, i.e., $\mathbf{g}{f,k}$ and $\mathbf{g}{b,k}$, are modeled as line-of-sight (LoS) channels \cite{Zhenyao2023}.
The transmit/receiver array steering vectors to the direction $\theta_k$ are thus denoted as
\begin{eqnarray}
    \mathbf{b}(\theta_k)= \sqrt{\frac{\zeta_b}{B}} \left[1, e^{j\pi \sin(\theta_k)}, \ldots, e^{j\pi (B-1) \sin(\theta_k)} \right]^{\rm{T}},
\end{eqnarray}
where $\mathbf{b} \in \{\mathbf{g}_{f,k}, \mathbf{g}_{b,k}\}$, $B\in\{M, N\}$, $\theta_k$ is $T_k$'s direction for the BS and the reader direction, and $\zeta_b$ denotes the path-loss. Finally, $\mathbf{G}_{\rm{SI}} \in  \mathbb{C}^{M\times N}$ is the SI channel matrix between the transmitter and the receiver antennas of the BS and is modeled as a Rician fading channel with a Rician factor of  $K_{\rm{SI}}$ \cite{Mohammadi2023, Diluka2024CFFD, Kim2021}. 

CSI is essential for  BS beamforming, user/reader data decoding, and other operations. To acquire CSI, the proposed system is assumed to utilize the time-division duplexing mode. That means coherence time is divided into two slots, one designated for channel estimation and one for data transmission.  In the first slot, pilot signals are transmitted,  which can yield high-quality CSI. Passive tags, however, have restricted pilot transmission capabilities \cite{ZhangZLK199}. Instead, they can modulate pilot sequences into the BS signal and reflect them to the reader, allowing the channel estimation via least square (LS)/MMSE estimate and deep leaning  \cite{Zargari10320395,rezaei2023timespread}. Although not the focus of this study, backscatter CSI estimation is complex, and recent results  \cite{rezaei2023timespread} are promising. Thus, full  CSI availability is a reasonable assumption. 

\begin{figure}[!t]
    \centering 
    \def\svgwidth{220pt} 
    \fontsize{8}{8}\selectfont 
    \graphicspath{{Figures/}}
    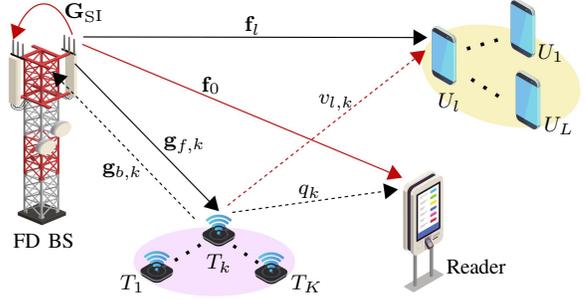  
    \caption{An RSMA-assisted ISABC system setup. Red lines denote interference links.} \label{fig_SystemModel}\vspace{-0mm} 
\end{figure}

\subsection{Transmission Model}
The  BS serves the primary users and simultaneously receives the tag-reflected signals. It uses RSMA to send messages to the primary users. This means that the message $d_l$ of $U_l$ is split into a common part $d_{c,l}$ and a private part $d_{p,l}$, i.e., $d_l = \{ d_{c,l}, d_{p,l}\}$ for $l \in \mathcal{L}$. The common parts of all users, i.e., $\{d_{c,1}, \ldots, d_{c,L} \}$ are jointly encoded into the common stream $x_c$ while the private parts, i.e., $\{d_{p,1}, \ldots, d_{p,L} \}$  are  encoded into private streams $\{x_{1}, \ldots, x_{L} \}$, which are assumed to be independent \cite{Clerckx2016, Galappaththige2024RSMA}. The common and private streams are then linearly precoded  using the precoders $\mathbf{w}_c \in \mathbb{C}^{M \times 1}$ and $\mathbf{w}_l \in \mathbb{C}^{M \times 1}$ for $l \in \mathcal{L}$, respectively \cite{Clerckx2016, Galappaththige2024RSMA}.

Conversely, the BS signal also includes  sensing signal $\mathbf{s} \in \mathbb{C}^{M \times 1}$ with the covariance matrix $\mathbf{S} \triangleq \mathbb{E}\{\mathbf{s} \mathbf{s}^{\rm{H}} \}$. This extends the degrees-of-freedom of the transmitted signal  $\mathbf{x}$ to achieve enhanced sensing performance \cite{Zhenyao2023}. The BS transmitted signal can thus  be expressed as
\begin{eqnarray}\label{Transmit_data}
    \mathbf{x} = \mathbf{w}_c x_c + \sum\nolimits_{j\in \mathcal{L}} \mathbf{w}_{j} x_j +\mathbf{s}.
\end{eqnarray}
In \eqref{Transmit_data},  $x_c$, $\{x_l\}_{l\in \mathcal{L}}$, and $\mathbf{s}$ are mutually independent. 

To backscatter data, tag $T_k$ uses a multi-level ($\tilde{M}$-ary)  modulation scheme \cite{Diluka2022, Rezaei2023Coding}. Tag data symbol $c_k$, with $\mathbb{E}\{\vert c_k \vert^2\} = 1$, is thus selected from a $\tilde{M}$-ary constant-envelope modulation. Consequently, the users receive the BS signal and the tags' backscattered signals. Assuming the propagation delay difference for all signals is negligible \cite{Liao2020}, the received signal at $U_l$ is given by
\begin{eqnarray}
    y_l' = \mathbf{f}_l^{\rm{H}} \mathbf{x} + \sum\nolimits_{k\in \mathcal{K}} \sqrt{\alpha_k} \mathbf{h}_{l,k}^{\rm{H}} \mathbf{x} c_k + z_u,
\end{eqnarray}
where the first and second terms are the direct-link signals (BS-to-user) and the backscatter-link signals (BS-to-tags-to-user), respectively, $z_u \sim \mathcal{CN}(0,\sigma^2)$ is the additive white Gaussian noise (AWGN) at $U_l$, with $0$ mean and $\sigma^2$ variance, and $\mathbf{h}_{l,k}$ is the effective backscatter channel at $U_l$ through $T_k$, i.e., $\mathbf{h}_{l,k}= \mathbf{g}_{f,k}(\theta_k) v_{l,k}$. Because primary users and BS are integral to the primary networks, a pre-existing connection permits the exchange of information, including sensing waveforms, via a control link \cite{positioningLTE}. It is thus assumed that the users are aware of the sensing waveform, however, partially eliminate it before decoding data. This is due to the limited dynamic range of the users, even if the channels are perfectly known. After the detecting signal is removed, the received signal is given as
\begin{eqnarray}\label{eqn_rx_user}
    y_l &=& \mathbf{f}_l^{\rm{H}} \mathbf{w}_c x_c + \sum\nolimits_{j\in\mathcal{L}} \mathbf{f}_l^{\rm{H}} \mathbf{w}_j x_j + \sqrt{\delta_s} \mathbf{f}_l^{\rm{H}} \mathbf{s} \nonumber \\
    &&\!\!\!\!\!\!\!\!\!\!\!\!\!\!\!\!\! + \sum\nolimits_{k\in \mathcal{K}} \sqrt{\alpha_k} \mathbf{h}_{l,k}^{\rm{H}} \bigg(\mathbf{w}_c x_c + \sum\nolimits_{j\in \mathcal{L}} \mathbf{w}_{j} x_j +\mathbf{s} \bigg) c_k + z_u, \quad
\end{eqnarray}
where $\delta_s \in [0,1]$ accounts for the SIC quality for the sensing signal. 

As the reader and the BS cooperate, the reader is aware of the BS transmitted signal. The reader thus applies SIC before decoding tags data. The post-processed received signal at the reader is given as
\begin{eqnarray}\label{eqn_rx_reader}
    y_r &=& \sqrt{\delta_c} \mathbf{f}_0^{\rm{H}} \mathbf{w}_c x_c + \sqrt{\delta_p} \sum\nolimits_{j\in\mathcal{L}} \mathbf{f}_0^{\rm{H}} \mathbf{w}_j x_j + \sqrt{\delta_s} \mathbf{f}_0^{\rm{H}} \mathbf{s}  \nonumber \\ 
    &&\!\!\!\!\!\!\!\!\!\!\!\!\!\!\!\! + \sum\nolimits_{k\in \mathcal{K}} \sqrt{\alpha_k} \mathbf{h}_{k}^{\rm{H}} \bigg(\mathbf{w}_c x_c + \sum\nolimits_{j\in \mathcal{L}} \mathbf{w}_{j} x_j +\mathbf{s} \bigg) c_k + z_r, \quad
\end{eqnarray}
where $\mathbf{h}_k = \mathbf{g}_{f,k}(\theta_k) q_{k}$ and $z_r \sim \mathcal{CN}(0,\sigma^2)$ is the AWGN at the reader. Moreover,  $\delta_c \in [0,1]$ and $\delta_p \in [0,1]$ are the SIC qualities for the common stream and private streams, respectively. 

Tag backscattered signals reach the users, reader, and BS. Thus, those signals enable the BS to retrieve environmental information  \cite{Diluka2023}. The BS received signal, i.e., $\mathbf{y}_b \in \mathbb{C}^{N\times 1}$, is given as
\begin{eqnarray}\label{eqn_rx_BS}
    \mathbf{y}_b = \sum\nolimits_{k\in \mathcal{K}} \sqrt{\alpha_k} \mathbf{G}_k (\theta_k) \mathbf{x} c_k + \mathbf{G}_{\rm{SI}}^{\rm{H}} \mathbf{x}  + \mathbf{z}_b,
\end{eqnarray}
where $\mathbf{G}_k (\theta_k) \triangleq \mathbf{g}_{b,k}(\theta_k) \mathbf{g}_{f,k}^{\rm{H}}(\theta_k)$,   $\sqrt{\alpha_k} \mathbf{G}_k (\theta_k) \mathbf{x} c_k$ is $T_k$'s ($k$-th target) reflection, and $\mathbf{z}_b \sim \mathcal{CN}(\mathbf{0},\sigma^2 \mathbf{I}_N)$ is the AWGN at the BS. In \eqref{eqn_rx_BS}, the second term is the SI at the BS receiver due to simultaneous transmission and reception. Due to the limited dynamic range of the receiver, SI cannot be canceled perfectly even with the perfect CSI of the SI channel \cite{Mohammadi2023, Diluka2024CFFD}. After the SI cancellation, the BS applies the receiver beamformer, $\mathbf{u}_k \in \mathbb{C}^{N\times 1}$ for $k \in \mathcal{K}$, to the received signal \eqref{eqn_rx_BS} to acquire the desired reflected signal of $T_k$. The post-processed signal for acquiring $T_k$ sensing information is thus given as
\begin{eqnarray}\label{eqn_T_kSens}
    {y}_{b,k} &=& \sqrt{\alpha_k} \mathbf{u}_k^{\rm{H}} \mathbf{G}_k (\theta_k) \mathbf{x} c_k + \sum\nolimits_{i\in \mathcal{K}_k} \sqrt{\alpha_i} \mathbf{u}_k^{\rm{H}} \mathbf{G}_i (\theta_i) \mathbf{x} c_i \nonumber \\
    &&+ \sqrt{\beta} \mathbf{u}_k^{\rm{H}} \mathbf{G}_{\rm{SI}}^{\rm{H}} \mathbf{x} + \mathbf{u}_k^{\rm{H}} \mathbf{z}_b,  
\end{eqnarray}
where $\mathcal{K}_k \triangleq \mathcal{K}\setminus\{k\}$ and $0<\beta \ll 1$ is a constant denoting the SI cancellation ability of the FD BS, i.e., the SI cancellation quality. FD radios can use SI cancellation techniques involving hardware techniques, analog domain cancellation, and digital domain cancellation of residual SI \cite{Mohammadi2023, Diluka2024CFFD}. These can also be adapted for ISAC/ISABC systems.

\subsection{Tags' EH}\label{sec_EH}
As mentioned, passive tags do not generate RF signals and rely entirely on EH to power their essential functions. To simultaneously perform EH and data transfer operations, tags split the incident RF signal \cite{Zhang2013, Hakimi2023} into two streams. With this PS, $T_k$ reflects a fraction of the incident RF  power, i.e., $p_{k}^{\rm{in}}=|\qg^{\rm{H}}_{f,k}\qw|^2+\qg^{\rm{H}}_{f,k}\qS\qg_{f,k}$, and harvests the remainder based on the reflection coefficient, $\alpha_k$ \cite{Zhang2013}. In particular, $T_k$ reflects $\alpha_k p_{k}^{\rm{in}}$ or data transmission, while absorbing the reminder, $(1-\alpha_k) p_{k}^{\rm{in}}$, for EH. 

The harvested power at $T_k$, $p_{k}^{\rm{h}}$, can be modeled as a linear or nonlinear function of $p_{k}^{\rm{in}}$. The harvested power of the widely used linear model is given as $p_{k}^{\rm{h}} = \eta(1-\alpha_k) p_{k}^{\rm{in}}$, where $\eta \in (0,1]$ is the power conversion efficiency. Despite its simplicity, the linear EH model overlooks nonlinear characteristics of actual EH circuits, such as saturation and sensitivity \cite{Boshkovska2015}. As a remedy, a parametric nonlinear EH model based on the sigmoid function has been frequently employed \cite{Boshkovska2015}. It models the total harvested power at $T_k$ as $p_{k}^{\rm{h}} = \Phi((1-\alpha_k)p_{k}^{\rm{in}})$, where $\Phi(\cdot)$ is a function representing non-linear effects \cite[eq. (4)]{Boshkovska2015}.

Nevertheless, our problem formulation can handle linear and non-linear models within one unified framework. The activation threshold, $p_b$, is an essential parameter whether a linear or non-linear model is used. It is the minimal amount of power necessary to activate the EH circuit, which is around  \qty{-20}{\dB m} for off-the-shelf passive tags \cite{Diluka2022}. This means that the harvested power must exceed  $p_b$ to activate the tag, i.e., $(1-\alpha_k) p_{k}^{\rm{in}} \geq  p_b'$, where $p_b' \triangleq \Phi^{-1}(p_b)$. The following utilizes the nonlinear EH model to formulate the optimization problem and design the resource allocation algorithm. 

\section{Communication and Sensing Performance}

This section derives the communication and sensing performance of the proposed RSMA-assisted ISABC system. 

\subsection{Primary Communication Performance}
The users first decode the common message by treating private and backscatter signals as interference. From \eqref{eqn_rx_user}, the rate of $U_l$ for decoding common data  is 
\begin{eqnarray}
    \mathcal{R}_{c,l} =  {\rm{log}}_2(1+\gamma_{c,l}),
\end{eqnarray}
where $\gamma_{c,l}$ is the received SINR at $U_l$ for decoding $x_c$, and given in  \eqref{SINR_ue_common}.
\begin{figure*}
\begin{eqnarray}\label{SINR_ue_common}
    \gamma_{c,l} = \frac{\vert \mathbf{f}_l^{\rm{H}} \mathbf{w}_c \vert^2} {\sum_{j \in \mathcal{L}}\vert \mathbf{f}_l^{\rm{H}} \mathbf{w}_j \vert^2 + \delta_s \vert \mathbf{f}_l^{\rm{H}} \mathbf{s} \vert^2 + \sum_{k \in \mathcal{K}} \alpha_k  \left(\vert \mathbf{h}_{l,k}^{\rm{H}}  \mathbf{w}_c \vert^2  + \sum_{j \in \mathcal{L}}\vert \mathbf{h}_{l,k}^{\rm{H}} \mathbf{w}_j \vert^2 + \vert \mathbf{h}_{l,k}^{\rm{H}}  \mathbf{s} \vert^2 \right)+ \sigma^2}
\end{eqnarray}

\hrulefill

\vspace{-0mm}

\end{figure*}
To guarantee that all users successfully decode the common stream, the rate of decoding the common stream $x_c$ should not exceed $\mathcal{R}_c = {\rm{min}}_{l \in \mathcal{L}}\mathcal{R}_{c,l}$ \cite{Xu2021}. As $\mathcal{R}_c$ is shared by $L$ users, it follows that  $\sum_{l \in \mathcal{L}} C_l = \mathcal{R}_c$, where $C_l$ is the portion of the common rate at $U_l$ transmitting $d_{c,l}$.

Next, $U_l$ subtracts the decoded common message from the received signal and decodes its desired private message, $x_{l}$. The private rate of $U_l$ is thus given by
\begin{eqnarray}
     \mathcal{R}_{p,l}' =  {\rm{log}}_2(1 + \gamma_{p,l}),
 \end{eqnarray}
 where $\gamma_{p,l}$ is the received SINR at $U_l$ for decoding $x_{l}$ and given in \eqref{SINR_ue_private}. In \eqref{SINR_ue_private} $\delta_c$ accounts for SIC qualities for removing the common signal. 
\begin{figure*}
\begin{eqnarray}\label{SINR_ue_private}
    \gamma_{p,l} = \frac{\vert \mathbf{f}_l^{\rm{H}} \mathbf{w}_l \vert^2}  {\delta_c \vert \mathbf{f}_l^{\rm{H}} \mathbf{w}_c \vert^2 + \sum_{j \in \mathcal{L}_l}\vert \mathbf{f}_l^{\rm{H}} \mathbf{w}_j \vert^2 +  \delta_s \vert \mathbf{f}_l^{\rm{H}} \mathbf{s} \vert^2 + \sum_{k \in \mathcal{K}} \alpha_k  \left( \vert \mathbf{h}_{l,k}^{\rm{H}}  \mathbf{w}_c \vert^2  + \sum_{j \in \mathcal{L}}\vert \mathbf{h}_{l,k}^{\rm{H}} \mathbf{w}_j \vert^2 + \vert \mathbf{h}_{l,k}^{\rm{H}}  \mathbf{s} \vert^2 \right)+ \sigma^2}
\end{eqnarray}

\vspace{-0mm}

\end{figure*}
The total achievable rate of $U_l$, including the portion of common rate transmitting $d_{c,l}$ and private rate transmitting $d_{p,l}$, is thus expressed as 
\begin{eqnarray}
    \mathcal{R}_{p,l} = C_l + \mathcal{R}_{p,l}'.
\end{eqnarray}

\subsection{\bc Performance}
The reader decodes the tags' data. Using \eqref{eqn_rx_reader}, the reader decodes $T_k$'s data at the rate 
\begin{eqnarray}
    \mathcal{R}_{t,k} \approx {\rm{log}}_2(1 + \gamma_{t,k}),
\end{eqnarray}
where $\gamma_{t,k}$ is the $T_k$'s SINR at the reader and given in \eqref{SINR_tag}.
\begin{figure*}
\begin{eqnarray}\label{SINR_tag}
    \gamma_{t,k} = \frac{\alpha_k  \left( \vert \mathbf{h}_{k}^{\rm{H}} \mathbf{w}_c\vert^2 + \sum_{j \in \mathcal{L}} \vert \mathbf{h}_{k}^{\rm{H}} \mathbf{w}_j\vert^2 + \vert \mathbf{h}_{k}^{\rm{H}} \mathbf{s} \vert^2 \right)}{ \delta_c \vert \mathbf{f}_0 \mathbf{w}_c \vert^2 + \delta_p \sum_{j \in \mathcal{L}} \vert \mathbf{f}_{0}^{\rm{H}} \mathbf{w}_j\vert^2 + \delta_s \vert \mathbf{f}_{0}^{\rm{H}} \mathbf{s} \vert^2 + \sum_{i \in \mathcal{K}_k} {\alpha_i} \left( \vert \mathbf{h}_{i}^{\rm{H}} \mathbf{w}_c\vert^2 + \sum_{j \in \mathcal{L}} \vert \mathbf{h}_{i}^{\rm{H}} \mathbf{w}_j\vert^2  + \vert \mathbf{h}_{i}^{\rm{H}} \mathbf{s}\vert^2 \right)+ \sigma^2 }
\end{eqnarray}

\vspace{-0mm}

\end{figure*}

\subsection{Sensing Performance}
Sensing performance is typically measured by the transmit beampattern gain or mean squared error of the transmit beampattern \cite{He2022, Stoica2007}. Although simple, they do not account for the receiver's beam pattern or multi-target interference. This may cause ambiguity in multi-target detection due to interference between the targets' reflected signals \cite{Zhenyao2023, Cui2014}. 
{Conversely, the CRB solely focuses on the lower bound of estimation error (i.e., accuracy) {\cite{Tang2019, Cui2014}}. CRB indicates how precisely a parameter (e.g., distance, velocity) can be estimated, but it doesn’t reflect how much information about the environment is being captured over time.} 
To overcome these limitations in such sensing measures, the sensing SINR has been suggested to assess sensing performance \cite{Zhenyao2023, Cui2014}.  The detection probability of a target (tag) is proportional to its sensing SINR \cite{Zhenyao2023, Cui2014}. In addition, the sensing SINR enables target detection through both transmit and receiver beamforming, reducing interference between targets \cite{Zhenyao2023, Cui2014}. Due to these benefits,   this study employs sensing SINR to measure sensing performance.

Sensing is done by the BS, which utilizes the backscattered signal of $T_k$. The BS applies the SI cancellation and the receiver beamformer, $\mathbf{u}_k$, to received signal \eqref{eqn_rx_BS} to capture $T_k$'s reflected signal. From  \eqref{eqn_T_kSens}, the sensing rate of $T_k$ is obtained as 
\begin{eqnarray}
    \mathcal{R}_{s,k} \approx {\rm{log}}_2(1 + \Upsilon_{k}),
\end{eqnarray}
where $\Upsilon_{k}$ is the sensing SINR of $T_k$ and given in \eqref{eqn_sens_SINR_tag}.
\begin{figure*}
\begin{eqnarray}\label{eqn_sens_SINR_tag}
    \Upsilon_{k}\!\!\!\!\! &=&\!\!\!\!\!\frac{\alpha_k \mathbb{E} \left\{ \vert \mathbf{u}_k^{\rm{H}} \mathbf{G}_k (\theta_k) \mathbf{x} \vert^2 \right\} }{\!\!\!\sum\limits_{i\in \mathcal{K}_k} \!\!\alpha_i \mathbb{E} \left\{ \vert \mathbf{u}_k^{\rm{H}} \mathbf{G}_i(\theta_i) \mathbf{x} \vert^2 \right\}\! +\! \beta \mathbb{E} \left\{ \vert \mathbf{u}_k^{\rm{H}} \mathbf{G}_{\rm{SI}}^{\rm{H}} \mathbf{x} \vert^2 \right\} \! + \!\mathbb{E} \left\{ \vert \mathbf{u}_k^{\rm{H}} \mathbf{z}_b \vert^2 \right\}   } \!= \!\frac{\alpha_k \mathbf{u}_k^{\rm{H}} \mathbf{G}_k (\theta_k) \mathbf{R}_x  \mathbf{G}_k^{\rm{H}} (\theta_k) \mathbf{u}_k  }{ \mathbf{u}_k^{\rm{H}}\! \left( \! \sum\limits_{i\in \mathcal{K}_k}\!\! \alpha_i  \mathbf{G}_i(\theta_i) \mathbf{R}_x \mathbf{G}_i^{\rm{H}} \!+\! \beta \mathbf{G}_{\rm{SI}} \mathbf{R}_x   \mathbf{G}_{\rm{SI}}^{\rm{H}}  \!+ \!\sigma^2 \mathbf{I}_N  \right) \!\mathbf{u}_k  }  \qquad
\end{eqnarray}
\hrulefill

\vspace{-0mm}

\end{figure*}
$\mathbf{R}_x \triangleq \mathbb{E} \{\mathbf{x} \mathbf{x}^{\rm{H}} \} = \mathbf{w}_c\mathbf{w}_c^{\rm{H}} + \sum_{j\in \mathcal{L}} \mathbf{w}_j\mathbf{w}_j^{\rm{H}} + \mathbf{S}$ in \eqref{eqn_sens_SINR_tag} is the BS transmitted signal covariance matrix \cite{Zhenyao2023}.

\begin{rem}
{The sensing rate, measured in bps/Hz, refers to the rate at which useful information is acquired from the environment through sensing {\cite{Tang2019, Zhenyao2023, Cui2014, Diluka2023}}, i.e., it measures how effectively the system gathers information about the target, such as location, velocity, or other properties. Conversely, the sensing rate is a comprehensive, real-time measure of sensing performance that is consistent with communication metrics, i.e., rate, making it ideal for optimizing the trade-off between sensing and communication {\cite{Tang2019, Zhenyao2023, Cui2014, Diluka2023}}. Additionally, it also aids in enhancing the performance of specific sensing metrics such as the CRB, which assesses the precision of parameter estimates but may not always be the primary concern in practical ISAC/ISABC deployments.} 
\end{rem}

\section{Problem Formulation}
Herein, the BS transmit power of the proposed RSMA-assisted ISABC network is optimized.  In particular, our objective is to minimize it by jointly optimizing the BS receive beamforming, $\{\qu_k\}_{k\in \mathcal{K}}$, the BS transmit beamforming, $\qw_c$, $\{\qw_l\}_{l\in \mathcal{L}}$, covariance matrix  $ \qS$, the tag reflection coefficients,  $\{\alpha_k\}_{k\in \mathcal{K}}$, and common rates of the users, $\{C_l\}_{l\in \mathcal{L}}$. The set of all the  optimization variables is denoted  as $\mathcal{A} = \left\{ \{\qu_k\}_{k\in \mathcal{K}}, \qw_c, \{\qw_l\}_{l\in \mathcal{L}}, \qS \succeq 0,   \{\alpha_k\}_{k\in \mathcal{K}}, \{C_l\}_{l\in \mathcal{L}}  \right\}$. 

The communication rates for the users and the reader must exceed their thresholds. The sensing and EH functionalities of the tags must be met. The optimization problem is thus formulated as follows:
\begin{subequations}\label{P1_prob}
    \begin{eqnarray}
       \mathbf{P}_1:
       \underset {\mathcal{A}}{\rm{min}} &&    \Vert \mathbf{w}_c  \Vert^2+ \sum\nolimits_{l \in \mathcal{L}} \Vert \mathbf{w}_l  \Vert^2 + \tr{(\mathbf{S})}, \label{P1_obj} \\
       \text{s.t.} && \mathcal{R}_{s,k} \geq \mathcal{R}_{s,k}^{\rm{th}},  ~ \forall k, \label{P1_sens_rate}\\
       &&  \mathcal{R}_{t,k} \geq \mathcal{R}_{t,k}^{\rm{th}},  ~ \forall k, \label{P1_tag_Rate}\\
       && C_l + \mathcal{R}_{p,l}' \geq \mathcal{R}_{p,l}^{\rm{th}} , ~\forall  l, \label{P1_pvt_rate}\\
       && \sum\nolimits_{l \in \mathcal{L}} C_l \leq \mathcal{R}_{c}, \label{P1_ComRate_decod}\\
       &&  0 \leq C_l,  ~ \forall l, \label{P1_cl}\\ 
       &&  {p}_k^{\rm{in}} \geq \frac{\Phi^{-1}(p_b)}{1-\alpha_k}, ~ \forall k, \label{P1_tag_EH}\\
       &&  \mathbf{S}  \succeq 0, \label{P1_S} \\
       &&\Vert \mathbf{u}_k\Vert^2 = 1, ~ \forall k,  \label{P1_detector}\\
       &&  0 <\alpha_k < 1, ~ \forall  k , \qquad \label{P1_alpha}
    \end{eqnarray}
\end{subequations}
where constraints \eqref{P1_sens_rate} and \eqref{P1_tag_Rate} ensure the required sensing and communication rate requirement of each tag at the BS and the reader, respectively, with  $\mathcal{R}_{s,k}^{\rm{th}}$ and $\mathcal{R}_{t,k}^{\rm{th}}$ representing the corresponding targeted sensing and \bc rates. On the other hand, \eqref{P1_pvt_rate} sets the minimum rate requirements for the user to decode private data, in which $\mathcal{R}_{p,l}^{\rm{th}}$ denotes the targeted private rate of the respective user.  Conditions \eqref{P1_ComRate_decod} and \eqref{P1_cl} are necessary to guarantee the successful decoding of the common stream. Besides, \eqref{P1_tag_EH} guarantees the minimum tag power requirements for EH. Constraint \eqref{P1_detector} is the normalization constraint for the BS reception filter. Finally, \eqref{P1_alpha} is the range of the tag reflection coefficient.  

Since the user/tag communication/sensing rates are monotonically increasing functions of their arguments, i.e., SINRs, they can be replaced with respective SINRs before proceeding to the proposed solution. The optimization problem, $\mathbf{P}_1$, is thus equivalently formulated as follows: 
\begin{subequations}\label{P2_prob}
    \begin{eqnarray}
      \mathbf{P}_2:
       \underset {\mathcal{A}}{\rm{min}} &&   \Vert \mathbf{w}_c  \Vert^2+ \sum\nolimits_{l \in \mathcal{L}} \Vert \mathbf{w}_l  \Vert^2 + \tr{(\mathbf{S})}, \label{P2_obj} \\
       \text{s.t.} && \Upsilon_{k} \geq  \Upsilon_{k}^{\rm{th}},  ~ \forall k, \label{P2_sens_rate}\\
       &&  \gamma_{t,k} \geq \Gamma_{t,k}^{\rm{th}},  ~ \forall k, \label{P2_tag_Rate}\\
       && \eqref{P1_pvt_rate}-\eqref{P1_alpha}, \label{P2_18e}
    \end{eqnarray}
\end{subequations}
where $\Upsilon_{k}^{\rm{th}} \triangleq 2^{\mathcal{R}_{s,k}^{\rm{th}}} - 1$ and $\Gamma_{t,k}^{\rm{th}} \triangleq 2^{\mathcal{R}_{t,k}^{\rm{th}}} - 1$ are the respective sensing/communication SINR thresholds. 

\section{Proposed Solution}\label{pro_solu}
Since problem $\mathbf{P}_2$ is non-convex, widely available convex algorithms cannot be directly applied.  Thus, an AO strategy \cite{bezdek2003convergence}
is adopted to handle it, where the problem is decomposed into three sub-problems \cite{bertsekas1997nonlinear}. It turns out that the objective value of $ \mathbf{P}_2 $ is the same regardless of the solutions of the first and third sub-problems. Thus, they are feasibility problems. Feasibility problems aim to find any solution that satisfies the constraints. However, these are converted into optimization problems with explicit objectives for more effective solutions.

\subsection{Sub-Problem 1: Optimizing over $\mathbf{u}_k$}\label{Sec_receiver_combine}
This develops the optimal reception filters of the  BS  while keeping other variables constant. For given $\{\mathbf{w}_c, \{\mathbf{w}_l\}_{l\in \mathcal{L}}, \mathbf{S}, \{\alpha_k \}_{k\in \mathcal{K}}, \{C_l\}_{l\in \mathcal{L}}\}$, $\mathbf{P}_2$ becomes a feasibility problem. Therefore, any feasible value of $\qu_k$ that complies with \eqref{P2_sens_rate} and \eqref{P1_detector} can be a potential solution.

Although vector $\qu_k$ may not have a direct impact on reducing the BS transmit power, it is crucial for optimizing the sensing SINR at the BS for each tag. Our methodology prioritizes the optimization of $\qu_k$ to meet the sensing performance criteria and indirectly contribute to transmit power reduction. Enhanced SINR for tag signals reduces the need for higher transmit power, thus aligning with our overarching aim of power minimization \cite{Stanczak2008book, Wan2016}.

Using the unique structure of the sensing SINR for each tag \eqref{eqn_sens_SINR_tag},  this sub-problem is transformed into a generalized Rayleigh quotient problem, providing closed-form optimal combiner vectors \cite{Stanczak2008book}. To this end, $\mathbf{P}_2$ is reduced to the following optimization problem:
\begin{subequations}\label{Pu_prob}
    \begin{eqnarray}
        \mathbf{P}_{\mathrm{u}}:
        \underset { \mathbf{u}_k}{\text{max}} && \nonumber \\
        &&\!\!\!\!\!\!\!\!\!\!\!\!\!\!\!\!\!\!\!\!\!\!\!\!\!\!\!\!\! \frac{\alpha_k \qu_k^{\rm{H}} \qG_k \qR_x \qG^{\rm{H}}_k \qu_k}{\qu^{\rm{H}}_k\left(\sum\limits_{i\in \mathcal{K}_k} \alpha_i \qG_i \qR_x \qG_i^{\rm{H}} + \beta \qG_{\rm{SI}} \qR_x \qG_{\rm{SI}}^{\rm{H}}  +  \sigma^2_k \qI_N \right)\qu_k}, \label{Pu_obj} \qquad\\
        \text{s.t.} && \Vert \mathbf{u}_k\Vert^2 = 1, \quad \forall  k. \label{Pu_detector}
    \end{eqnarray}
\end{subequations}
Next, the objective function, \eqref{Pu_obj}, is rearranged to convert it into the following optimization problem:
\begin{subequations}\label{Pu_prob1}
    \begin{eqnarray}
        \mathbf{P}_{\mathrm{u}1}:
        \underset { \mathbf{u}_k}{\text{max}} && \frac{ \mathbf{u}_k^{\rm{H}}\tilde{\mathbf{G}}_{k} \tilde{\mathbf{G}}_{k}^{\rm{H}} \mathbf{u}_k}{ \mathbf{u}_k^{\rm{H}} \mathbf{Q}_k\mathbf{u}_k }, \label{Pu_obj1} \qquad\\
        \text{s.t.} && \eqref{Pu_detector},
    \end{eqnarray}
\end{subequations}
where $\mathbf{Q}_k \triangleq \sum_{i \in \mathcal{K}_k} \!{\alpha_i} \mathbf{G}_{i} \qR_x \qG_{i}^{\rm{H}} + \beta \mathbf{G}_{\rm{SI}}\mathbf{R}_x \mathbf{G}_{\rm{SI}}^{\rm{H}}+ \sigma^2 \mathbf{I}_N$, and $\tilde{\mathbf{G}}_{k} = \sqrt{\alpha_k} \mathbf{G}_{k} \left(\mathbf{w}_c + \sum_{j\in \mathcal{L}} \mathbf{w}_{j}  +\mathbf{s}  \right)$. Problem $\mathbf{P}_{\mathrm{u}1}$ thus becomes a generalized Rayleigh ratio quotient problem \cite{Stanczak2008book}. The optimal combiner vector for $T_k$ is characterized by the Wiener or MMSE filter \cite{Hakimi2023} and is given as  
\begin{eqnarray}\label{eqn_optimal_u}
    \mathbf{u}_k^* = \frac{\mathbf{Q}_k^{-1} \tilde{\mathbf{G}}_{k}}{\Vert \mathbf{Q}_k^{-1} \tilde{\mathbf{G}}_{k} \Vert}, \quad  \forall k. 
\end{eqnarray}

\subsection{Sub-Problem 2: Optimizing over $\mathbf{w}_c$, $\mathbf{w}_l$, $\mathbf{S}$, and $C_l$}
This  optimizes for the BS transmit beamforming, i.e., $\mathbf{w}_c$, $\{\mathbf{w}_l\}_{l\in\mathcal{L}}$, and $\mathbf{S}$, and users’ common rate, i.e., $\{C_l\}_{l\in \mathcal{L}}$ for given $\{\{\qu_k\}_{k\in \mathcal{K}}, \{\alpha\}_{k\in \mathcal{K}}\}$. However, owing to the interference within the SINRs, Constraints $\mathbf{P}_{2}$ become non-convex for the BS transmit beamformers. The semi-definite relaxation (SDR) method is employed to overcome this barrier. Thus, $\mathbf{P}_2$ is reformulated as
\begin{subequations}\label{Pw_prob}
    \begin{eqnarray}
       \mathbf{P}_{\mathrm{w}}:
       \underset {\mathbf{w}_c, \mathbf{w}_l, \mathbf{S}, C_l }{\rm{min}} &&    \Vert \mathbf{w}_c  \Vert^2+ \sum\nolimits_{l \in \mathcal{L}} \Vert \mathbf{w}_l  \Vert^2 + \tr{(\mathbf{S})}, \label{Pw_obj} \\   \text{s.t.}  && \eqref{P2_sens_rate}-\eqref{P2_tag_Rate},~ \eqref{P1_pvt_rate}-\eqref{P1_S}. \qquad
    \end{eqnarray}
\end{subequations}
$\mathbf{P}_{\mathrm{w}}$  can be effectively addressed  by using the SDR method \cite{so2007approximating, Qingqing2019}. Define  matrices $\mathbf{W}_c \triangleq \mathbf{w}_c \mathbf{w}_c^{\rm{H}}$ and $\mathbf{W}_l \triangleq \mathbf{w}_l \mathbf{w}_l^{\rm{H}}$ for $l\in \mathcal{L}$. Here, $\mathbf{W}_c$ and $\mathbf{W}_l$ are semi-definite matrices with rank one constraints, i.e., ${\rm{Rank}}(\mathbf{W}_c) = 1$ and ${\rm{Rank}}(\mathbf{W}_l) = 1$. By relaxing the highly non-convex rank one constraints, $\mathbf{P}_{\mathrm{w}}$ can be reformulated into a conventional semi-definite programming (SDP) problem. However, as constraints \eqref{P1_pvt_rate} and \eqref{P1_ComRate_decod} are still non-convex, they are next converted into tractable convex forms.  Define 
\begin{subequations}
\begin{align}
    A_l &\triangleq \tr(\qf_l \qf_l^{\rm{H}} \qW_l), \\
    B_l &\triangleq \delta_c \tr(\qf_l \qf_l^{\rm{H}} \qW_c) + \sum\nolimits_{j \in \mathcal{L}_l} \tr( \qf_l \qf_l^{\rm{H}} \qW_j) +  \delta_s \tr(\qf_l \qf_l^{\rm{H}} \qS) \nonumber\\
    &+ \sum\nolimits_{k \in \mathcal{K}} \alpha_k  \bigg( \tr(\qh_{l,k} \qh_{l,k}^{\rm{H}} \qW_c )  + \sum\nolimits_{j \in \mathcal{L}} \tr(\qh_{l,k} \qh_{l,k}^{\rm{H}} \qW_j)  \nonumber\\
    &\hspace{30mm}+ \tr(\qh_{l,k} \qh_{l,k}^{\rm{H}} \qS) \bigg) + \sigma^2 .
\end{align}
\end{subequations}
The SINR in \eqref{P1_pvt_rate} is now bounded such that $0 \leq \mu_l \leq A_l/B_l$, where $\mu_l$ is an auxiliary variable. Thereby introducing another auxiliary variable, $\phi_l$, where  $A_l \geq \mu_l \phi_l$ and $B_l \leq \phi_l$. However, since the product $\mu_l \phi_l$ is non-convex,  Taylor-series linearization is used to convert it to a convex function as 
\begin{eqnarray}
    \mu_l \phi_l = \mu_l^{(t)} \phi_l^{(t)} + \phi_l^{(t)}(\mu_l - \mu_l^{(t)}) + \mu_l^{(t)} (\phi_l-\phi_l^{(t)}) \triangleq v_l,\!
\end{eqnarray}
where $\{\mu_l,\phi_l\}$ and $\{\mu_l^{(t)},\phi_l^{(t)}\}$ are the current and the previous iteration values of $\{\mu_l,\phi_l\}$, respectively. The constraint \eqref{P1_pvt_rate} is thus equivalent to 
\begin{subnumcases} {}
C_l + \log_2(1+ \mu_l) \geq \mathcal{R}_{p, l}^{\rm{th}}, & $\forall l$,\label{Pw_pvt_rate_c1}\\
A_l \geq v_l, & $\forall l$, \label{Pw_pvt_rate_c2}\\
B_l \leq \phi_l, & $\forall l$. \label{Pw_pvt_rate_c3}
\end{subnumcases}
For constraint \eqref{P1_ComRate_decod}, each rate value $\mathcal{R}_{c}$ is replaced with $\sum_{l \in \mathcal{L}} C_l \leq \mathcal{R}_{c,l}$ for $l\in\mathcal{L}$. Since $\mathcal{R}_{c,l}$ is not a convex function of the optimization variables, the SCA method is employed to linearize the constraint as 
\begin{eqnarray}\label{eqn_ComRate_decod_c1}
    \sum\nolimits_{l \in \mathcal{L}} C_l \leq \digamma_l - \tilde{\digamma}_l,
\end{eqnarray}
where $\digamma_l$ and $\tilde{\digamma}_l$ are given in \eqref{eqn_digammal} and \eqref{eqn_digammal_tild}, respectively, 
\begin{figure*}[!t]
\begin{eqnarray}
    \digamma_l &\triangleq& \log_2  \bigg( \tr(\qf_l \qf_l^{\rm{H}} \qW_c) + \sum_{j \in \mathcal{L}} \tr(\qf_l \qf_l^{\rm{H}} \qW_j ) + \delta_s \tr( \qf_l \qf_l^{\rm{H}} \qS)  \nonumber\\
    && + \sum_{k \in \mathcal{K}} \alpha_k   \bigg( \tr( \qh_{l,k} \qh_{l,k}^{\rm{H}}  \qW_c ) + \sum_{j \in \mathcal{L}} \tr(\qh_{l,k}  \qh_{l,k}^{\rm{H}} \qW_j ) + \tr( \qh_{l,k} \qh_{l,k}^{\rm{H}} \qS)   \bigg)+ \sigma^2  \bigg) \label{eqn_digammal}\\
    \tilde{\digamma}_l &\triangleq&  \log_2(\Lambda_l) + \frac{\sum_{k\in \mathcal{K}} \alpha_k \tr(\qh_{l,k} \qh_{l,k}^{\rm{H}}) }{\ln(2) \Lambda_l} \tr\left(\qW_c - \qW_c^{(t)}\right) + \sum_{j\in\mathcal{L}} \frac{ \tr(\qf_{l} \qf_{l}^{\rm{H}}) + \sum_{k\in \mathcal{K}} \alpha_k \tr(\qh_{l,k} \qh_{l,k}^{\rm{H}}) }{\ln(2) \Lambda_l} \tr\left(\qW_j - \qW_j^{(t)}\right) \nonumber \\
    && \frac{ \delta_s \tr(\qf_{l} \qf_{l}^{\rm{H}}) + \sum_{k\in \mathcal{K}} \alpha_k \tr(\qh_{l,k} \qh_{l,k}^{\rm{H}}) }{\ln(2) \Lambda_l} \tr\left(\qS - \qS^{(t)}\right) \label{eqn_digammal_tild}
\end{eqnarray}



\vspace{-0mm}

\end{figure*}
and $\Lambda_l$ is defined as
\begin{eqnarray} \label{eqn_lambda}
    \Lambda_l &\triangleq& \sum_{j \in \mathcal{L}} \tr \left(\qf_l \qf_l^{\rm{H}} \qW_j^{(t)} \right) + \delta_s \tr\left( \qf_l \qf_l^{\rm{H}} \qS^{(t)}\right) \nonumber \\
    &&\!\!\!\!\!\!\!\!\!\!\!\!\! + \sum_{k \in \mathcal{K}} \alpha_k   \bigg( \tr\left( \qh_{l,k} \qh_{l,k}^{\rm{H}}  \qW_c^{(t)} \right) + \sum_{j \in \mathcal{L}} \tr\left(\qh_{l,k}  \qh_{l,k}^{\rm{H}} \qW_j^{(t)} \right)  \nonumber \\
    &&\!\!\!\!\!\!\!\!\!\!\!\!\!   + \tr\left(\qh_{l,k} \qh_{l,k}^{\rm{H}} \qS^{(t)}\right)   \bigg)+ \sigma^2. 
\end{eqnarray}
In \eqref{eqn_lambda}, $(\cdot)^{(t)}$ denotes the previous iteration values of respective variables. Finally, $\mathbf{P}_{\mathrm{w}}$ can be reformulated to a relaxed SDP problem as given in \eqref{Pw1_prob}, which can be tackled using the CVX tool \cite{boyd2004convex, grant2014cvx}.  
\begin{figure*}
\begin{small}
\begin{subequations}\label{Pw1_prob}
    \begin{eqnarray}
       \!\!\!\!\!\!\!\!\! \mathbf{P}_{\mathrm{w}1}: 
        && \!\!\!\!\!\!\!\!
         \underset {\mathbf{W}_c, \mathbf{W}_l,\mathbf{S}, C_l }{\rm{min}}  \tr{(\mathbf{W}_c)} + \sum_{l \in \mathcal{L}} \tr{(\mathbf{W}_l)} +  \tr{(\mathbf{S})}, \label{Pw1_obj} \\
        \text{s.t.}
         && \!\!\!\!\!\!\!\! \Upsilon_{k}^{\rm{th}} \qu_k^{\rm{H}} \left(\sum_{i \in \mathcal{K}_k} \alpha_i \left(  \tr( \qG_i^{\rm{H}} \qG_i \qW_c)   + \sum_{j \in \mathcal{L}} \tr(\qG_i^{\rm{H}} \qG_i \qW_j)  + \tr(\qG_i^{\rm{H}} \qG_i \qS) \right) + \sigma^2 \mathbf{I}_N   \right. \nonumber \\     
         && \!\!\!\!\!\!\!\! \left. +\beta \left( \tr( \qG_{\rm{SI}}^{\rm{H}} \qG_{\rm{SI}} \qW_c) + \sum_{j \in \mathcal{L}} \tr(\qG_{\rm{SI}}^{\rm{H}} \qG_{\rm{SI}} \qW_j)  + \tr(\qG_{\rm{SI}}^{\rm{H}} \qG_{\rm{SI}} \qS) \right)\right)\qu_k \nonumber \\
         && \!\!\!\!\!\!\!\! - \alpha_k \qu_k^{\rm{H}} \left(  \tr( \qG_k^{\rm{H}} \qG_k \qW_c)   + \sum_{j \in \mathcal{L}} \tr(\qG_k^{\rm{H}} \qG_k \qW_j)  + \tr(\qG_k^{\rm{H}} \qG_k \qS) \right) \qu_k \leq 0, ~ \forall k,  \\
         &&\!\!\!\!\!\!\!\! \Gamma_{t,k}^{\rm{th}} \left( \delta_c \tr(\qf_0 \qf_0^{\rm{H}} \mathbf{W}_c) + \delta_p \sum_{j \in \mathcal{L}}\tr(\qf_0 \qf_0^{\rm{H}}  \mathbf{W}_j) + \delta_s \tr(\qf_0 \qf_0^{\rm{H}} \qS) + \sigma^2 + 
         \sum_{i \in \mathcal{K}_k} \alpha_i   \left( \tr(\qh_i \qh_i^{\rm{H}} \mathbf{W}_c) + \sum_{j \in \mathcal{L}}\tr(\qh_i \qh_i^{\rm{H}}  \mathbf{W}_j) \right. \right. \nonumber \\
         &&\!\!\!\!\!\!\!\! \left.\left.  + \tr(\qh_i \qh_i^{\rm{H}} \mathbf{S}) \right)   \right) - \alpha_k \left( \tr(\qh_k \qh_k^{\rm{H}} \mathbf{W}_c) + \sum_{j \in \mathcal{L}}\tr(\qh_k \qh_k^{\rm{H}}  \mathbf{W}_j) + \tr(\qh_k \qh_k^{\rm{H}} \mathbf{S}) \right) \leq 0, ~ \forall k,  \\
        &&\!\!\!\!\!\!\!\! P_{\rm{th}}-(1-\alpha_k) \bigg(\tr(\qg_{f,k}\qg^{\rm{H}}_{f,k}\qW_c) + \sum_{j\in \mathcal{L}} \tr(\qg_{f,k}\qg^{\rm{H}}_{f,k}\qW_j)+\tr(\qg^{\rm{H}}_{f,k}\qg_{f,k}\qS)  \bigg)\leq 0, ~ \forall k, \\
        &&\!\!\!\!\!\!\!\! \{\mathbf{W}_c, \{\mathbf{W}_l\}_{l \in \mathcal{L}},  \mathbf{S}  \}\succeq 0, \\
        &&\!\!\!\!\!\!\!\! \eqref{P1_cl},~ \eqref{Pw_pvt_rate_c1}-\eqref{Pw_pvt_rate_c3},~ \eqref{eqn_ComRate_decod_c1}.
    \end{eqnarray}
\end{subequations}
\end{small}

\vspace{-0mm}

\hrulefill

\vspace{-0mm}

\end{figure*}

If the solutions to the relaxed problem $\mathbf{P}_{\mathrm{w}1}$ are of rank one,  the optimal transmit beamformers are obtained by eigenvalue decomposition \cite{Qingqing2019}. Let the eigenvalue decomposition of $\qW$ (i.e., $\qW \in \{\qW_c^*, \{\qW_l^*\}_{l\in \mathcal{L}}\}$) to be $\qW = \qU \boldsymbol{\Sigma} \qU^{\rm{H}}$, where $\qU$ is a unitary matrix and $\boldsymbol{\Sigma}  = \text{diag}(\lambda_1, \dots, \lambda_{M})$ is a diagonal matrix. If $\qW^*$ is rank one, the optimal transmit beamformer, $\qw^*$, is the eigenvector for the maximum eigenvalue. If not,  the Gaussian randomization accounts for the relaxed rank-one constraint \cite{Qingqing2019}. Specifically, a  solution for $\mathbf{P}_{\mathrm{w}}$ is generated as $\bar{\qW} = \qU \boldsymbol{\Sigma} ^{1/2} \qr$, with $\qr \in \mathcal{CN}(0, \qI_{M})$. Such solutions are tested for  $10^5$ times,  and the best one is selected.  These numerous random realizations of  $\qr$ with the SDR technique ensure an $\frac{\pi}{4}$-approximation to the optimal value of $\mathbf{P}_{\mathrm{w}}$ \cite{so2007approximating, Qingqing2019}.

\subsection{Sub-Problem 3: Optimizing over $\alpha_k$}\label{sec_alpha_opt}
This final sub-problem requires  optimizing $\{\alpha\}_{k\in \mathcal{K}}$ for given $\{\mathbf{w}_c, \{\mathbf{w}_l\}_{l\in \mathcal{L}}, \mathbf{S}, \{\qu_k \}_{k\in \mathcal{K}}, \{C_l\}_{l\in \mathcal{L}}\}$. Note that objective  (18a) is not a function of the reflection coefficients of the tags. Thus, Sub-Problem 3 reduces to a constraint-satisfaction problem of $\mathbf{P}_{2}$. That is, any $\{\alpha_k\}_{k\in \mathcal{K}}$ that satisfies the constraints \eqref{P2_sens_rate}--\eqref{P2_18e}
is considered to be an optimal solution. However, to effectively handle Sub-Problem 3, a slack-optimization-based problem is proposed \cite{Shayan2021}. While satisfying the other constraints,  two new slack variables are introduced to optimize further the tags' SINR and EH margins, i.e., $t_1$ and $t_2$, respectively \cite{Shayan2021}. The proposed problem is thus given in \eqref{Palpha_prob}, 
\begin{figure*}
\begin{subequations}\label{Palpha_prob}
    \begin{eqnarray}
         \mathbf{P}_{\alpha}:&& \!\!\!\!
         \underset {\alpha_k, t_1, t_2}{\rm{min}}  \lambda_1 t_1 +\lambda_2 t_2 , \label{P1_obj_alpha} \\
       \!\!\! \text{s.t.} && \!\!\!\! \alpha_k \qu_k^{\rm{H}} \qG_k \qR_x \qG^{\rm{H}}_k \qu_k \geq \Upsilon_k^{\thr} {\qu^{\rm{H}}_k\left(\sum_{i\in \mathcal{K}_k} \alpha_i \qG_i \qR_x \qG_i^{\rm{H}} + \beta \qG_{\rm{SI}} \qR_x \qG_{\rm{SI}}^{\rm{H}} + \sigma^2 \qI_N \right)\qu_k}, ~ \forall k, \\
       && \!\!\!\! \alpha_k \left(\vert \qh_k^{\rm{H}} \mathbf{w}_c \vert^2 + \sum_{j \in \mathcal{L}} \vert \qh_{k}^{\rm{H}}  \mathbf{w}_i \vert^2 + \qh_k^{\rm{H}} \qS \qh_k  \right) \geq \Gamma_k^{\rm{th}} \left( \delta_c \vert \qf_0^{\rm{H}} \qw_c \vert^2 + \delta_p \sum_{j\in \mathcal{L}} \vert \qf_0^{\rm{H}} \qw_j \vert^2 + \delta_c \qf_0^{\rm{H}} \qS \qf_0 \right. \nonumber \\
       && \!\!\!\! \left. + \sum_{i \in \mathcal{K}_k} {\alpha_i} \left( \vert \mathbf{h}_{i}^{\rm{H}} \mathbf{w}_c\vert^2 + \sum_{j \in \mathcal{L}}\vert \mathbf{h}_{i}^{\rm{H}}  \mathbf{w}_j \vert^2 + \mathbf{h}_{i}^{\rm{H}} \mathbf{S} \mathbf{h}_{i}  \right) + \sigma^2 \right) + t_1, ~ \forall k, \\
       && \!\!\!\! \vert \mathbf{f}_l^{\rm{H}} \mathbf{w}_c \vert^2 \geq  \tilde{\Gamma}_{c,l}^{\rm{th}}  \left( \sum_{j \in \mathcal{L}} \vert \mathbf{f}_l^{\rm{H}} \mathbf{w}_j \vert^2 + \delta_s \vert \mathbf{f}_l^{\rm{H}} \mathbf{s} \vert^2 +  \sum_{k \in \mathcal{K}}  \alpha_k \left(\vert \mathbf{h}_{l,k}^{\rm{H}}  \mathbf{w}_c \vert^2  + \sum_{j \in \mathcal{L}}\vert \mathbf{h}_{l,k}^{\rm{H}} \mathbf{w}_j \vert^2 + \vert \mathbf{h}_{l,k}^{\rm{H}}  \mathbf{s} \vert^2 \right)  + \sigma^2 \right), ~  \forall l, \quad  \label{Palpha_user_Crate}\\
       && \!\!\!\! \vert \mathbf{f}_l^{\rm{H}} \mathbf{w}_l \vert^2\! \geq  \!\tilde{\Gamma}_{p,l}^{\rm{th}} \!\! \left(\! \delta_c \vert \mathbf{f}_l^{\rm{H}} \mathbf{w}_c \vert^2 \!+\!\! \sum_{j \in \mathcal{L}_l} \vert \mathbf{f}_l^{\rm{H}} \mathbf{w}_j \vert^2 \!+\!  \delta_s \vert \mathbf{f}_l^{\rm{H}} \mathbf{s} \vert^2  \! + \!\!\sum_{k \in \mathcal{K}}\! \!\alpha_k \!\! \left(\! \vert \mathbf{h}_{l,k}^{\rm{H}}  \mathbf{w}_c \vert^2  \!+\! \sum_{j \in \mathcal{L}}\vert \mathbf{h}_{l,k}^{\rm{H}} \mathbf{w}_j \vert^2 \!+\! \vert \mathbf{h}_{l,k}^{\rm{H}}  \mathbf{s} \vert^2 \!\!\right)+\!\! \sigma^2 \! \right), ~\forall l, \\
        && \!\!\!\! (1-\alpha_k)p_{k}^{\rm{in}}  \geq  {\Phi^{-1}(p_b)} + t_2, ~ \forall k, \\
        && \!\!\!  0 <\alpha_k < 1, ~ \forall k,
    \end{eqnarray}
\end{subequations}

\vspace{-0mm}

\hrulefill

\vspace{-0mm}

\end{figure*}
where $\tilde{\Gamma}_{c,l}^{\rm{th}} \triangleq 2^{\sum_{l\in \mathcal{L}} C_l} - 1$ and $\tilde{\Gamma}_{p,l}^{\rm{th}} \triangleq 2^{\mathcal{R}_{p,l}^{\rm{th}} - C_l} -1$. In \eqref{P1_obj_alpha}, $\lambda_1$ and $\lambda_2$ are  positive constants. $\mathbf{P}_{\alpha}$ is convex, and thus can be efficiently solved by solvers such as CVX \cite{grant2014cvx}. Introducing slack variables in $\mathbf{P}_{\alpha}$ converts strict constraints into adjustable ones with a definable margin. This facilitates the convergence process by setting a more tangible minimization goal and aligns well with the convergence strategies of iterative solvers like CVX due to the explicit objective guiding the solution path \cite{Shayan2021, Qingqing2019}.

Algorithm \ref{alg_AO} presents the overall steps to solve  problem $\mathbf{P}_1$ \eqref{P1_prob}. It begins by setting $\{\{\alpha_k\}_{k\in \mathcal{K}}, \mathbf{w}_c, \mathbf{w}_l, \mathbf{S}, C_l\}$ to random feasible values that satisfy the constraints in $\mathbf{P}_1$, and then refines the received/transmit beamforming, reflection coefficients, and common rate iteratively until the normalized reduction in the total transmit power is less than $\epsilon = 10^{-3}$.

Variable initialization assumes a crucial role in determining algorithm convergence. Closer initialization to optimal values typically accelerates convergence. However, complex optimization problems with numerous constraints pose challenges to finding a near-optimal starting guess, potentially escalating overall computing complexity \cite{boyd2004convex}. Researchers often resort to randomly generating feasible initial values \cite{Qingqing2019, Shen2018}, a strategy also embraced in this study. Nonetheless, simulation results show that our overall and individual algorithms achieve rapid convergence regardless of the random initial values.

\begin{algorithm}[!t]
\caption{Overall Algorithm}
\begin{algorithmic}[1]
\label{alg_AO}
\STATE \textbf{Input}: Set the iteration counter $t = 0$, the convergence tolerance $\epsilon > 0$, initial feasible solution $\{\{\alpha_k\}_{k\in \mathcal{K}}, \mathbf{w}_c, \mathbf{w}_l, \mathbf{S}, C_l\}$. Initialize the objective function value $F^{(0)} = 0$.  
\WHILE{ $ \frac{F^{(t+1)} - F^{(t)}}{F^{(t+1)}} \geq \epsilon$}
\STATE Solve \eqref{eqn_optimal_u} to derive optimal received beamformer, $\qu_k^{(t+1)}$.
\STATE Solve $\mathbf{P}_{\mathrm{w}1}$ \eqref{Pw1_prob} to obtain suboptimal transmit beamformers, $\{\mathbf{w}^{(t+1)}_c, \mathbf{w}^{(t+1)}_l, \mathbf{S}^{(t+1)}, C^{(t+1)}_l\}$ by applying Gaussian randomization to recover the rank-one solution.
\STATE Solve $\mathbf{P}_{\alpha}$ \eqref{Palpha_prob} to obtain the suboptimal power reflection coefficients, $\alpha_k^{(t+1)}$. 
\STATE Calculate the objective function value $F^{(t+1)}$.
\STATE Set $t\leftarrow t+1$;
\ENDWHILE
\STATE \textbf{Output}: Optimal solutions $\mathcal{A}^*$.
\end{algorithmic}
\end{algorithm}

\begin{rem}
This study utilizes the AO method, and each associated subproblem yields a local solution. The convergence principle of the AO method is well-documented \cite{bezdek2003convergence}. Simply put, if the individual sub-problems reach convergence, then the entire optimization process will converge as well \cite{bezdek2003convergence}. This study employs the SDR and slack-optimization techniques to determine $\{\mathbf{w}_c, \mathbf{w}_l, \mathbf{S}, C_l\}$ and $\{\alpha_k\}_{k\in \mathcal{K}}$, respectively. Meanwhile, $\{\qu_k\}_{k\in \mathcal{K}}$ is directly solved via the Rayleigh ratio quotient method. Both the SDR and slack-optimization are established methods with guaranteed convergence \cite{Qingqing2019}, reinforcing the convergence of our AO approach. Our simulation outcomes further validate this assertion (Fig.~\ref{fig_convergence}).
\end{rem}

\begin{theorem}
{\textbf{Algorithm Convergence:}
Algorithm~{\ref{alg_AO}} iterations yield a non-increasing sequence of objective values with guaranteed convergence.}
\end{theorem}

\begin{proof}
   {Please see  Appendix {\ref{Theorem_1}}.}
\end{proof}

\subsection{Complexity of Proposed Algorithm}
{The computational complexity of Algorithm {\ref{alg_AO}}  is analyzed.}

\subsubsection{Optimization over $\qu_k$}
{Here, the Rayleigh quotient method provides the optimal closed-form received beamformers. It requires $\mathcal{O}(N^3)$ complexity for computing the inverse of the matrix $\qQ$ while the MMSE filters for the $K$ tags {\eqref{eqn_optimal_u}} need complexity of $\mathcal{O}(KN^2)$. The total complexity for this problem is thus $\mathcal{O}(KN^2 + N^3)$.}

\subsubsection{Optimization over $\{\mathbf{w}_c, \mathbf{w}_l, \mathbf{S}, C_l\}$} 
{This SDR sub-problem is tackled using SDP, specifically the interior-point method. From {\cite[Th. 3.12]{polik2010interior}}, the order of computational complexity for a SDP problem with $m$ SDP constraints with an $n \times n$ positive semi-definite (PSD) matrix is given by $\mathcal{O}\left( \sqrt{n} \log\left({1}/{\epsilon}\right) (mn^3 + m^2n^2 + m^3) \right)$, where $\epsilon > 0$ is the solution accuracy. In our proposed problem, $\mathbf{P}_{\mathrm{w}}$ {\eqref{Pw_prob}}, $n = M$ and $m = 3(K + L) +2$, and the approximate computational complexity is $\mathcal{O}\left( (K+L)M^3 \sqrt{M} \log\left({1}/{\epsilon}\right)   \right)$.}

\subsubsection{Optimization over $\alpha_k$}
{The difference of convex method and the interior point method are employed {\cite{Shayan2021}}. The number of iterations for convergence is given by $\frac{\left({\log(C)}/{t^0\delta}\right)}{\log \epsilon}$, where $C$ is the overall number of constraints, $t^0$ signifies the initial approximation for the interior point method's accuracy, and $0 < \delta \ll 1$ is the stopping criterion {\cite{boyd2004convex}}.}

\subsubsection{Algorithm \ref{alg_AO}} 
{Finally, the proposed Algorithm {\ref{alg_AO}}'s computational complexity can be expressed asymptotically as $\mathcal{O} \!\left(\! I_o \! \left(\! (K \!+\! N) N^2 \!+\! (K\!+\!L)M^3 \sqrt{M} \log\!\left(\!\frac{1}{\epsilon}\!\right) \!+\!  \frac{\left({\log(C)}/{t^0\delta}\right)}{\log \epsilon} \!\right) \!\right)$ where $I_o$ is the overall number of iterations for the Algorithm  {\ref{alg_AO}} to converge. Despite its higher-order polynomial time complexity, the suggested approach exhibits noteworthy real-world performance for datasets up to a certain size, particularly when $N$ and $M$ are kept below a predefined threshold. For large datasets, adopting optimization measures such as parallel processing can dramatically enhance performance {\cite{leiserson2010parallel}}.}

\section{Simulation Results}
{Next, simulation results are presented to assess the performance of the proposed ISABC network. 
The 3GPP urban micro (UMi) model is chosen to depict the path-loss values {$\zeta_{a}, \zeta_b$}, operating at a frequency of ${f_c=\qty{3}{\GHz}}$ {\cite[Table B.1.2.1]{3GPP2010}}. Additionally, the AWGN variance, symbolized as {$\sigma^2$}, is modeled as $\sigma^2=10\log_{10}(N_0 B N_f)$ {\qty{}{\dB m}}. In this equation, $N_0$ is given by $\qty{-174}{\dB m/\Hz}$, while $B$ represents the bandwidth, and $N_f$ denotes the noise figure. Table~{\ref{table_notations}} lists the main simulation parameters unless mentioned otherwise. All simulations are evaluated for {\num{e3}} iterations.} 

{Considering smart home/city scenarios, the BS and the mobile reader are positioned at coordinates $\{0,0\}$ and $\{12,0\}$, respectively. Meanwhile, tags are sporadically placed within a circle whose center is $\{6,-4\}$ and which has a radius of {\qty{3}{\m}} {\cite{Galappaththige2023}}. Furthermore, the users are randomly scattered within a circle with a radius of {\qty{5}{\m}} and a center of $\{55,0\}$.}

\begin{table}[t]
    \centering
    \renewcommand{\arraystretch}{1.1}
    \setlength{\tabcolsep}{10pt} 
    \caption{{Simulation parameters.}}
    \begin{tabular}{|c |c |c |c |}
    \hline
    \textbf{Parameter} & \textbf{Value} 	& \textbf{Parameter} 	& \textbf{Value}   \\ \hline \hline
    $f_c$ 							& \qty{3}{\GHz}  	& $\mathcal{R}_{t,k}^{\rm{th}}$ 	& \qty{1}{bps/\Hz}  \\ \hline
    $B$  							& \qty{10}{\MHz}  	& $ \mathcal{R}_{p,k}^{\rm{th}}$  		& \qty{2}{bps/\Hz}  \\ \hline
    $N_f$ 							& \qty{10}{\dB}  	& $K_{\rm{SI}}$  		& \qty{3}{\dB}  \\ \hline
    $M=N$ 							& \num{8}   		& $\beta$ 								& \qty{-90}{\dB}   \\ \hline
    $L$ 							& \num{3}  			& $\{\delta_c, \delta_p, \delta_s \}$ 	& \qty{-10}{\dB}   \\ \hline
    $K$ 							& \num{2} 			& $p_b$  								& \qty{-20}{\dB m} \\ \hline
    $\mathcal{R}_{s,k}^{\rm{th}}$ 	& \qty{1}{bps/\Hz}	&  $\epsilon$										&  $10^{-3}$ \\ 
    \hline
    \end{tabular}
    \label{table_notations}
\end{table}

For comparative  evaluation purposes,  the following benchmark schemes are considered:

\subsubsection{NOMA-assisted ISABC} This benchmark uses NOMA to serve the primary users. For a fair comparison, \bc and sensing at the BS are also considered, i.e., $\mathbf{x} = \sum_{j \in \mathcal{L}} \mathbf{w}_{j} x_j + \qs$. Without loss of generality, the channel gains are sorted in descending order, i.e., $\Vert {\mathbf{f}}_1 \Vert^2 \ge \ldots \ge \Vert {\mathbf{f}}_L\Vert^2$. Thus, $U_l$  first decodes the data for users with high channel gains and applies SIC before decoding its data signal. The SINR of $U_l$ with NOMA is given in \eqref{SINR_NOMA},
\begin{figure*}
\begin{eqnarray}\label{SINR_NOMA}
    \gamma_{l} = \frac{\vert \mathbf{f}_l^{\rm{H}} \mathbf{w}_l \vert^2} {\delta_p \sum_{j \in \mathcal{L}_l'}\vert \mathbf{f}_l^{\rm{H}} \mathbf{w}_j \vert^2 + \sum_{j \in \mathcal{L}_l''}\vert \mathbf{f}_l^{\rm{H}} \mathbf{w}_j \vert^2 + \delta_s \vert \mathbf{f}_l^{\rm{H}} \mathbf{s} \vert^2 + \sum_{k \in \mathcal{K}} \alpha_k  \left( \sum_{j \in \mathcal{L}}\vert \mathbf{h}_{l,k}^{\rm{H}} \mathbf{w}_j \vert^2 + \vert \mathbf{h}_{l,k}^{\rm{H}}  \mathbf{s} \vert^2 \right)+ \sigma^2}
\end{eqnarray}

\vspace{-0mm}

\hrulefill 

\vspace{-0mm}

\end{figure*}
where $\mathcal{L}_l' \triangleq \{1,\ldots,l-1\}$ and $\mathcal{L}_l'' \triangleq \{l+1,\ldots,L\}$. Moreover, the tags' SINRs and sensing SINRs can be easily derived by removing the common beamforming in  \eqref{SINR_tag} and \eqref{eqn_sens_SINR_tag}, respectively. This benchmark facilitates the investigation of the impacts of RSMA compared to conventional NOMA.

\subsubsection{RSMA-assisted \bc}
This benchmark considers multiple RSMA-enabled primary users and multiple tags.
However, the BS has no sensing capabilities, i.e., $\mathbf{x} = \qw_c x_c + \sum_{j \in \mathcal{L}} \mathbf{w}_{j} x_j$. This benchmark helps to evaluate the cost of adding sensing.

\subsubsection{Conventional \bc}
This benchmark considers a conventional \abc system with $K$-tags and $L$-users. Specifically, the BS does not perform rate-splitting and sensing and uses SDMA to serve the users. Thus, the BS transmit signal becomes  $\mathbf{x} = \sum_{j \in \mathcal{L}} \mathbf{w}_{j} x_j$. These tags operate solely on  EH and send data to the reader by reflecting  $\mathbf x$.  This benchmark helps to evaluate the cost of incorporating sensing and rate-splitting. 

\subsubsection{Conventional ISABC}
This benchmark uses  SDMA to serve the primary users. The BS  senses environmental information from $K$ tags. Thus,  the BS transmit signal becomes  $\mathbf{x} = \sum_{j \in \mathcal{L}} \mathbf{w}_{j} x_j + \qs$. The benchmark helps to gauge the difference between  SDMA and RSMA.  

\subsubsection{Sensing-only scheme}
This benchmark scheme evaluates a system with $K$-targets/tags without primary and backscatter communications. The targets/tags use EH, allowing for a fair comparison. For example, targets may be wireless sensors/tags that do not send data but harvest energy to remain active.

These benchmark schemes are special cases of Algorithm \ref{alg_AO}, which can easily handle them with minor modifications. 

\subsection{Convergence Rate of Algorithm \ref{alg_AO}}

\begin{figure}[!t]\vspace{-0mm}	
\centering
\includegraphics[width=0.45\textwidth]{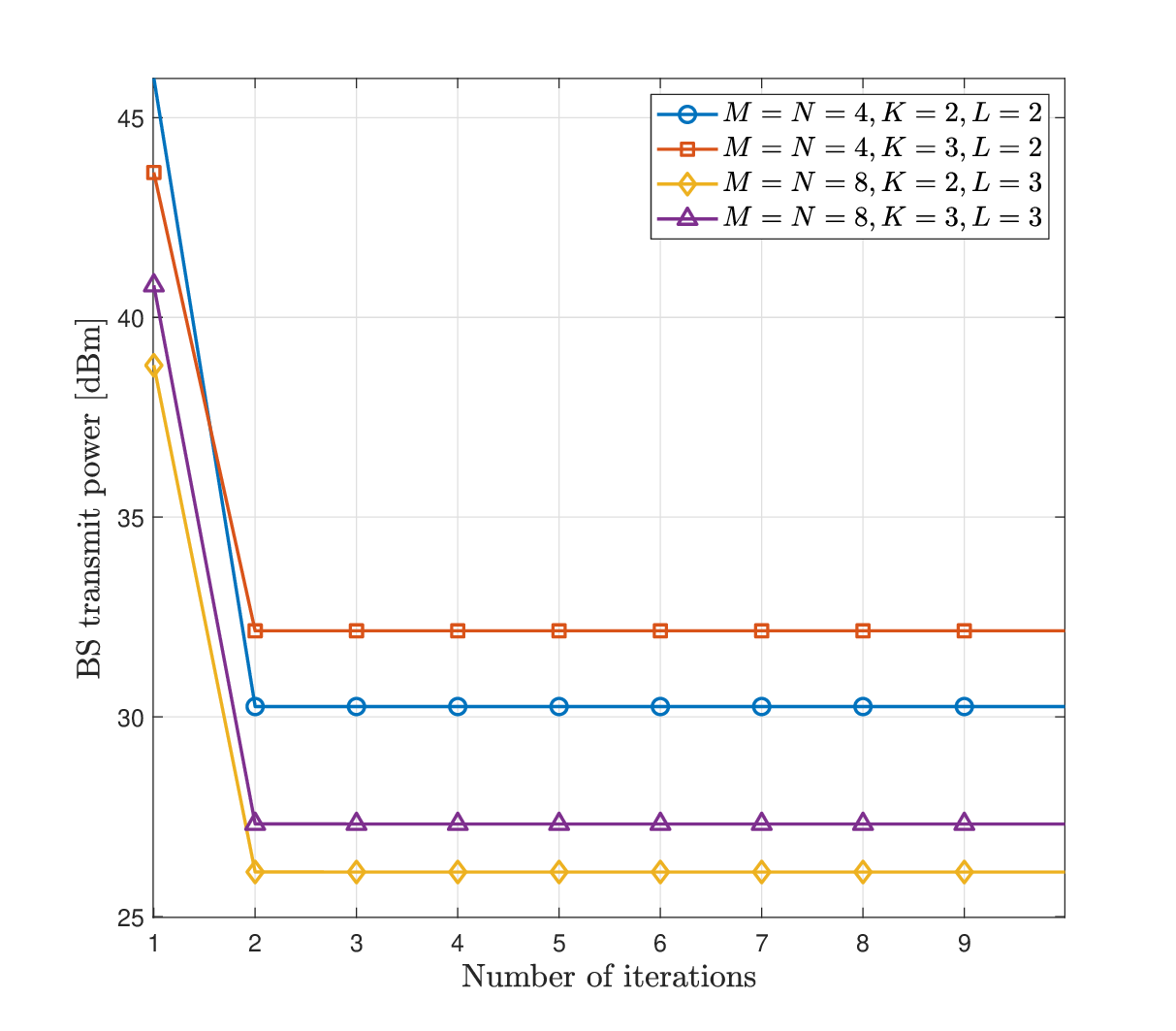}\vspace{-0mm}
    \caption{Convergence of the proposed Algorithm \ref{alg_AO}.}
	\label{fig_convergence} \vspace{-0mm}
\end{figure}

{Algorithm {\ref{alg_AO}}  outputs all variable set $\mathcal A$ in several iterations.   Thus, the saturation of the BS transmit power is a convergence test, and  Fig.~{\ref{fig_convergence}} thus plots it as a function of the number of iterations. The stopping criterion is that the normalized objective function increases less than $\epsilon = 10^{-3}$. The  BS transmit power decreases rapidly with each iteration and saturates after approximately three iterations, regardless of  $M$, $L$, or $K$. This implies quick convergence and validates the proposed algorithm's effectiveness. Extensive simulations reveal that Algorithm {\ref{alg_AO}} requires only three iterations to attain adequate performance with any system setup, hardly improving performance beyond three iterations.}

\subsection{Algorithm Execution Time}
\begin{figure}[!t]\vspace{-0mm}	
\centering
\includegraphics[width=0.45\textwidth]{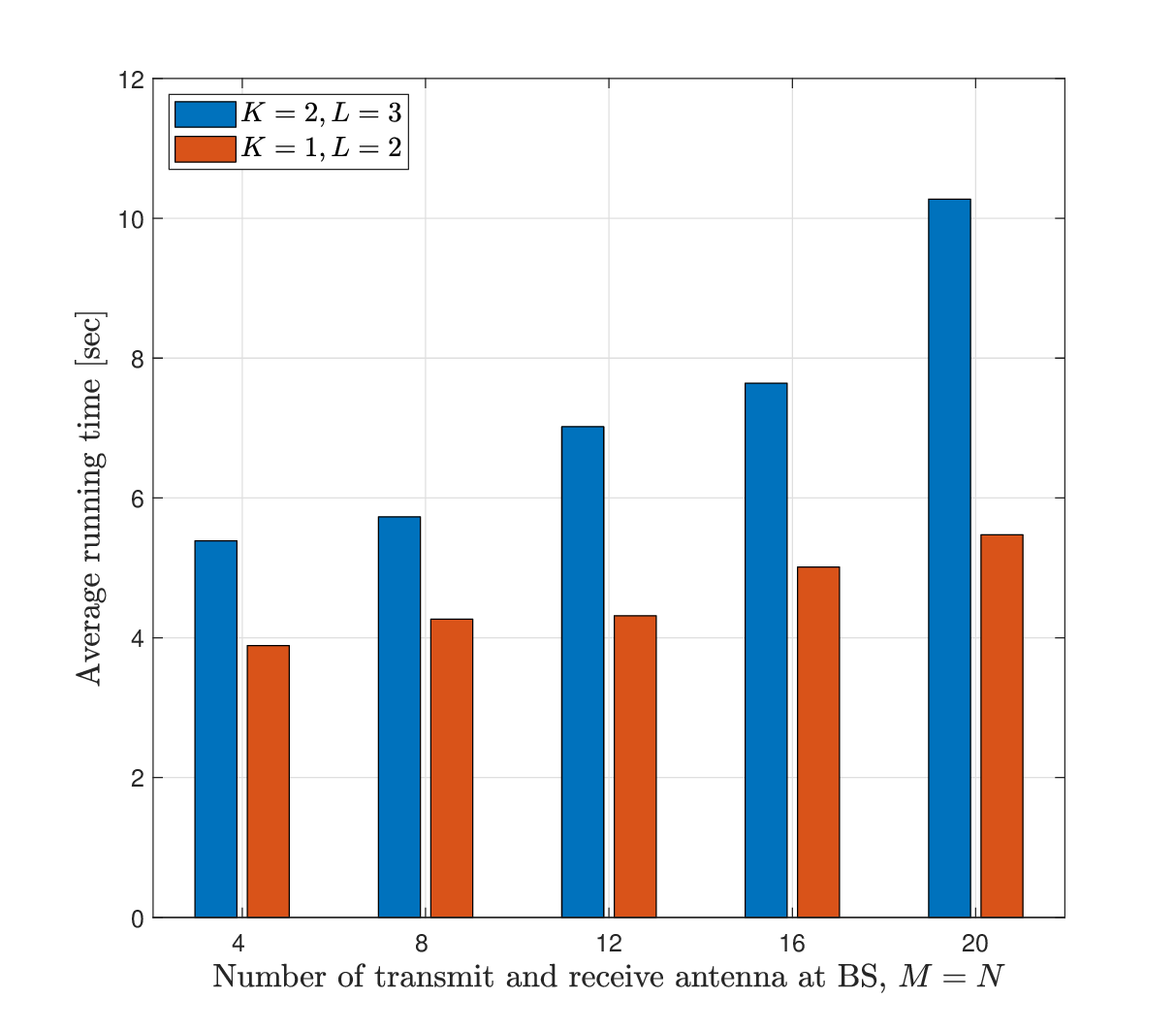}\vspace{-0mm}
    \caption{Average running time of the proposed Algorithm \ref{alg_AO} as a function of the number of BS transmit/receive antennas, $M=N$.}
	\label{fig_RunTime_M} \vspace{-0mm}
\end{figure}

{Fig.~{\ref{fig_RunTime_M}} shows the average runtime of Algorithm~{\ref{alg_AO}} as a function of the number of BS transmit/receive antennas ($M=N$), based on Matlab simulations using an Intel\textsuperscript{\textregistered} Core\textsuperscript{\texttrademark} i7 processor clocking at {\qty{2.8}{\GHz}}. The results indicate that the runtime increases proportionally with the number of BS transmit/receive antennas, highlighting the growing computational complexity as $M/N$ increases. Conversely, the run time is also proportional to the number of tags and/or users. This trend emphasizes the complex challenges of handling more BS antennas and highlights the significance of effective algorithms.}

{This study mainly focuses on application scenarios in low-mobility device or tag deployment contexts, such as smart homes and warehouses, which represent emerging backscatter networks as described in the 3GPP standard {\cite{ Huawei_ambient, Huawei}}. In such settings, channels have long coherence times due to the low mobility {\cite{Huawei_ambient, Huawei}}. Additionally, the algorithm's execution time depends largely on the BS's hardware and computational capacity, which typically includes dedicated hardware. As a result, the actual runtime is expected to be much shorter than depicted in Fig.~{\ref{fig_RunTime_M}}, remaining well within the coherence time.}

\subsection{Beampattern Gains}

\begin{figure}[!t]\vspace{-0mm}	
\centering
\includegraphics[width=0.45\textwidth]{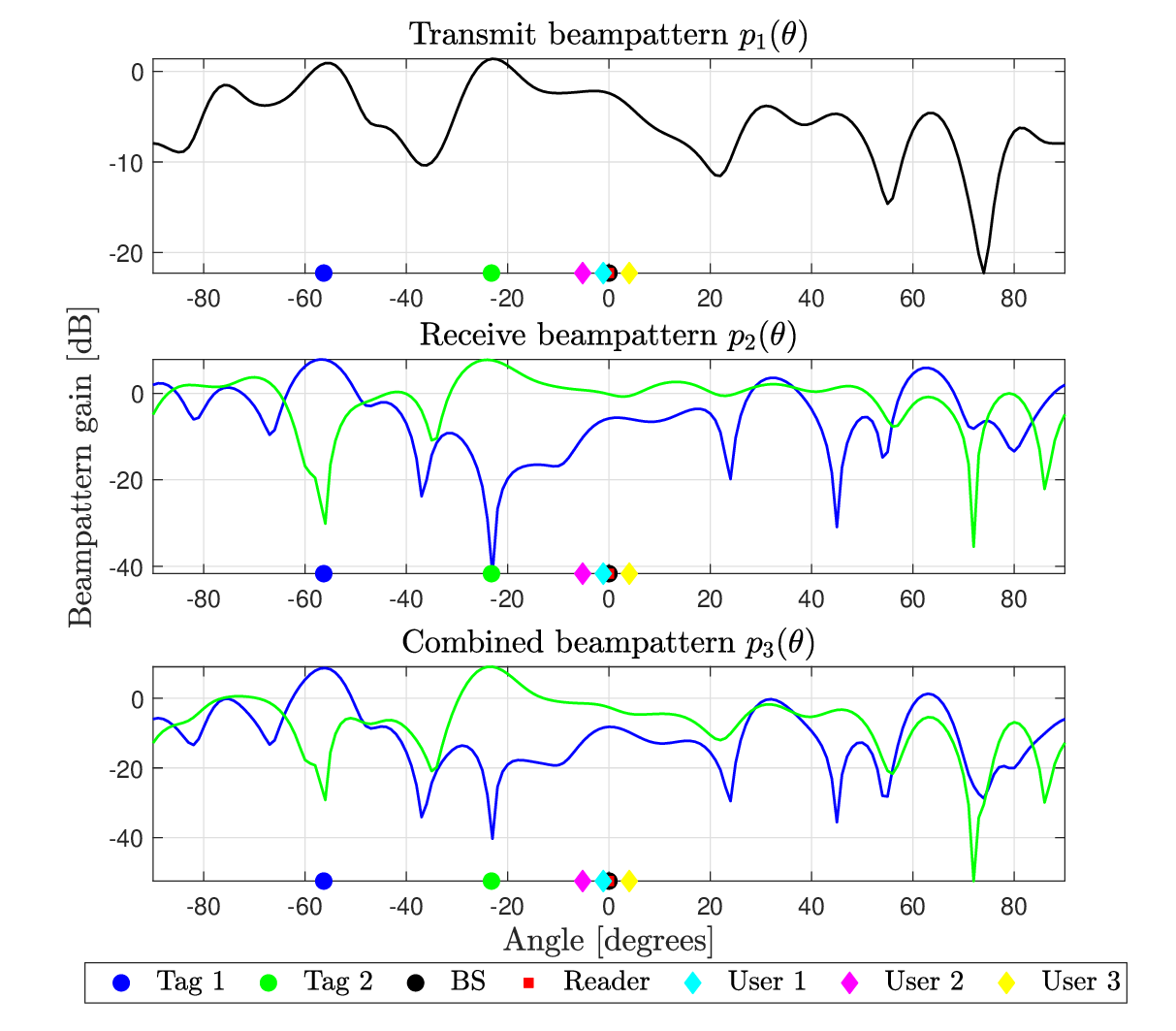}\vspace{-0mm}
    \caption{Beampattern of radar functionality of Algorithm \ref{alg_AO}.}
	\label{fig_BeamGain_proposed} \vspace{-0mm}
\end{figure}

\begin{figure}[!t]\vspace{-0mm}	
\centering
\includegraphics[width=0.45\textwidth]{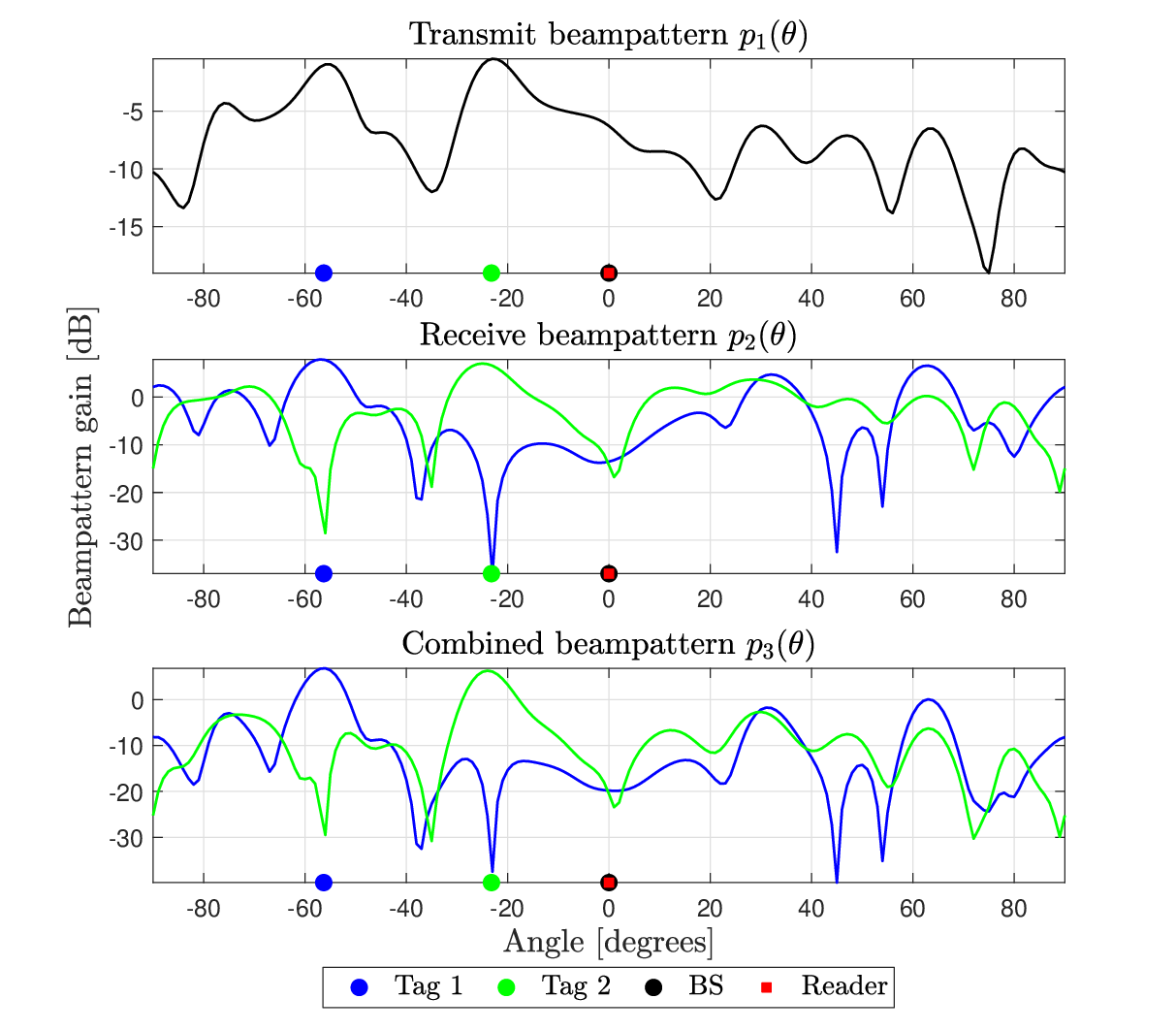}\vspace{-0mm}
    \caption{Beampattern of radar functionality of sensing-only scheme.}
	\label{fig_BeamGain_Sens_EH} \vspace{-0mm}
\end{figure}

The direct transmission and reception of beams in precise directions help to extract more information.  To this end, Algorithm \ref{alg_AO} functions as a linchpin, controlling the formation and direction of these beams. Beamforming refers to merging  signals from an array of antennas to produce a directed “beam” or “lobe.” This beam can be electronically steered while the antennas stay stationary. This steering feature improves signal quality, increases backscatter tag detection, and reduces possible interference significantly. 

The BS transmit signal, $\mathbf x,$ represents the outward-projected energy to illuminate the targets. The received beamformer is designed for clear reception, catching echoes or reflections off the backscatter tags. The  three crucial beampatterns are 
\begin{subequations}
\begin{eqnarray}
p_1(\theta) &=& \left| \mathbf{b}^{\rm{H}}(\theta_k) \qx^* \right|^2, \label{eq_p1} \\ 
p_2(\theta) &=& \left| (\qu_k^*)^{\rm{H}} \mathbf{b}(\theta_k) \right|^2, \label{eq_p2} \\
p_3(\theta) &=& \left| (\qu_k^*)^{\rm{H}} \mathbf{b}(\theta_k) \mathbf{b}^{\rm{H}}(\theta_k) \qx^* \right|^2. \label{eq_p3} 
\end{eqnarray}
\end{subequations}
Here, \eqref{eq_p1} indicates how transmitted energy disperses as a function of angle $\theta$, while \eqref{eq_p2} captures the sensitivity across multiple angles while receiving reflected energy. Moreover, \eqref{eq_p3} provides a composite representation incorporating the impacts of transmission and subsequent reflection processing.

Fig.~\ref{fig_BeamGain_proposed} and Fig.~\ref{fig_BeamGain_Sens_EH} illustrate the aforementioned three beampatterns generated using Algorithm \ref{alg_AO} and the sensing-only approach, respectively. The first sub-figure in Fig. \ref{fig_BeamGain_proposed} shows that Algorithm \ref{alg_AO} aims the transmit beams at tags and users. In contrast, the second sub-figure depicts the directions of interference from other tags for a specific tag with relatively deep nulls, separating the individual sensing information at the BS. The overall beampattern (third sub-figure) combines the transmit and receive beampatterns. Compared to the sensing-only scheme (Fig.~\ref{fig_BeamGain_Sens_EH}), Algorithm \ref{alg_AO} additionally provides users with communication capabilities while alleviating interference between targets/tags for sensing. 

Fig.~\ref{fig_BeamGain_proposed} and Fig.~\ref{fig_BeamGain_Sens_EH} help sense operations as they depict the radar/beamgain patterns. These figures offer a practical perspective, helping to better understand the intricacies and efficacy of beamforming algorithms.

\subsection{Effects of Number of BS Antennas}

\begin{figure}[!t]\vspace{-0mm}	
\centering
\includegraphics[width=0.45\textwidth]{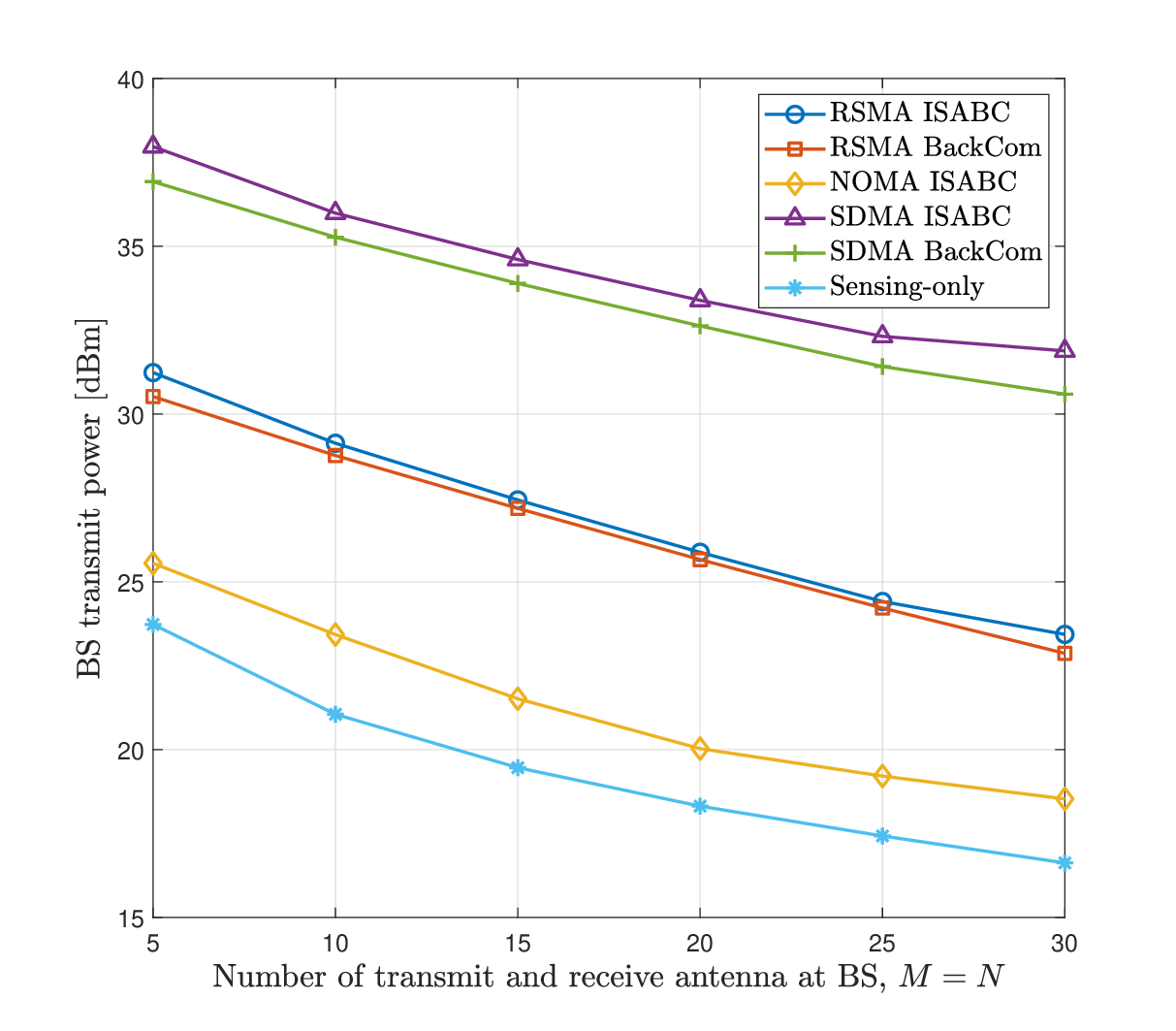}\vspace{-0mm}
    \caption{BS transmit power versus the number of BS transmit/receive antennas, $M=N$.} 
	\label{fig_TxPower_M} \vspace{-0mm}
\end{figure}

\begin{figure}[!t]\vspace{-0mm}	
\centering
\includegraphics[width=0.45\textwidth]{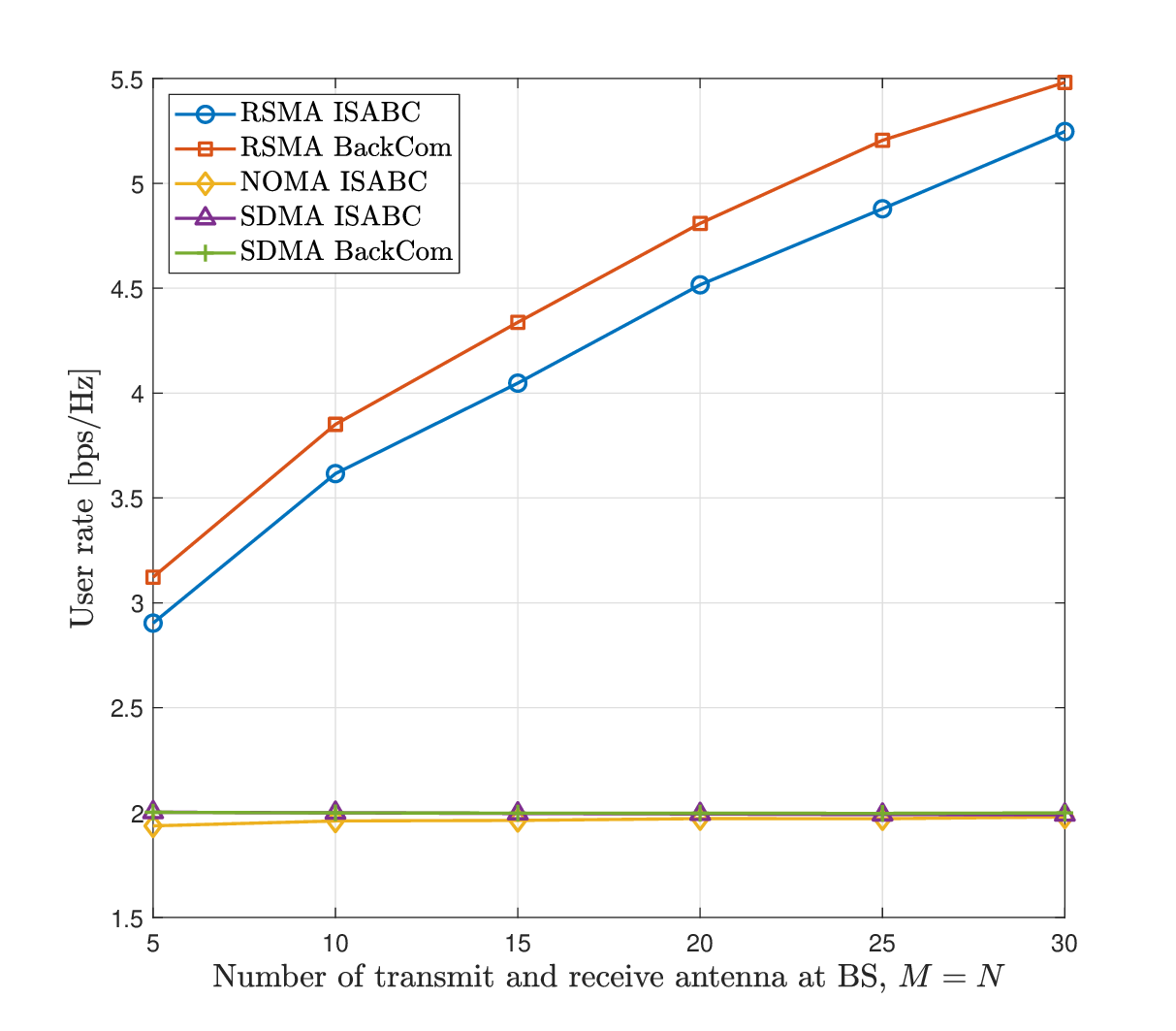}\vspace{-0mm}
    \caption{Average user rate versus the number of BS transmit/receive antennas, $M=N$.}
	\label{fig_UserRate_M} \vspace{-0mm}
\end{figure}

Fig.~\ref{fig_TxPower_M} and Fig.~\ref{fig_UserRate_M} show the influence of the number of BS antennas, $M=N$, on the BS transmit power and primary user rates, respectively. Fig.~\ref{fig_TxPower_M} shows that increasing $M$ has an inverse effect on the required power. Specifically, increasing  $M$ decreases the transmit power for all the schemes. This trend suggests the spatial diversity benefits for ISABC, namely  (1) increased communication rates and (2) significant power savings.  Additionally, NOMA and sensing-only schemes consume the lowest amount of transmit power. This is because NOMA has no common rate, and the sensing-only approach lacks communication. As a result, NOMA achieves a significantly lower user rate than RSMA (Fig~\ref{fig_UserRate_M}).
In contrast, conventional \bc and ISABC consume the most power resulting from multi-user multi-tag interference. Conversely, the proposed setup and RSMA-assisted \bc have moderate transmit power requirements. Importantly, the proposed IMBO algorithm enables sensing, a critical feature for IoT networks, with only a modest transmit power increase. 

Fig.~\ref{fig_UserRate_M} compares the average primary user rates of all the schemes, i.e., $\bar{\mathcal{R}}_{c} \triangleq 1/L \sum_{l\in \mathcal{L}}\mathcal{R}_{c,l}$ and  $\bar{\mathcal{R}}_{p} \triangleq 1/L \sum_{l\in \mathcal{L}}\mathcal{R}_{p,l}$. The figure shows the proposed scheme achieves the highest communication rate among all. The reason is that RSMA enables the optimization of the common rates, which increases them.   For example, just for a  $\qty{24}{\percent}$ transmit power increase ($M=N=\num{10}$), our scheme achieves a  $\qty{84.5}{\percent}$ rate gain over  NOMA-assisted ISABC.

\subsection{Transmit Power Versus Number of Tags}

\begin{figure}[!t]\vspace{-0mm}	
\centering
\includegraphics[width=0.45\textwidth]{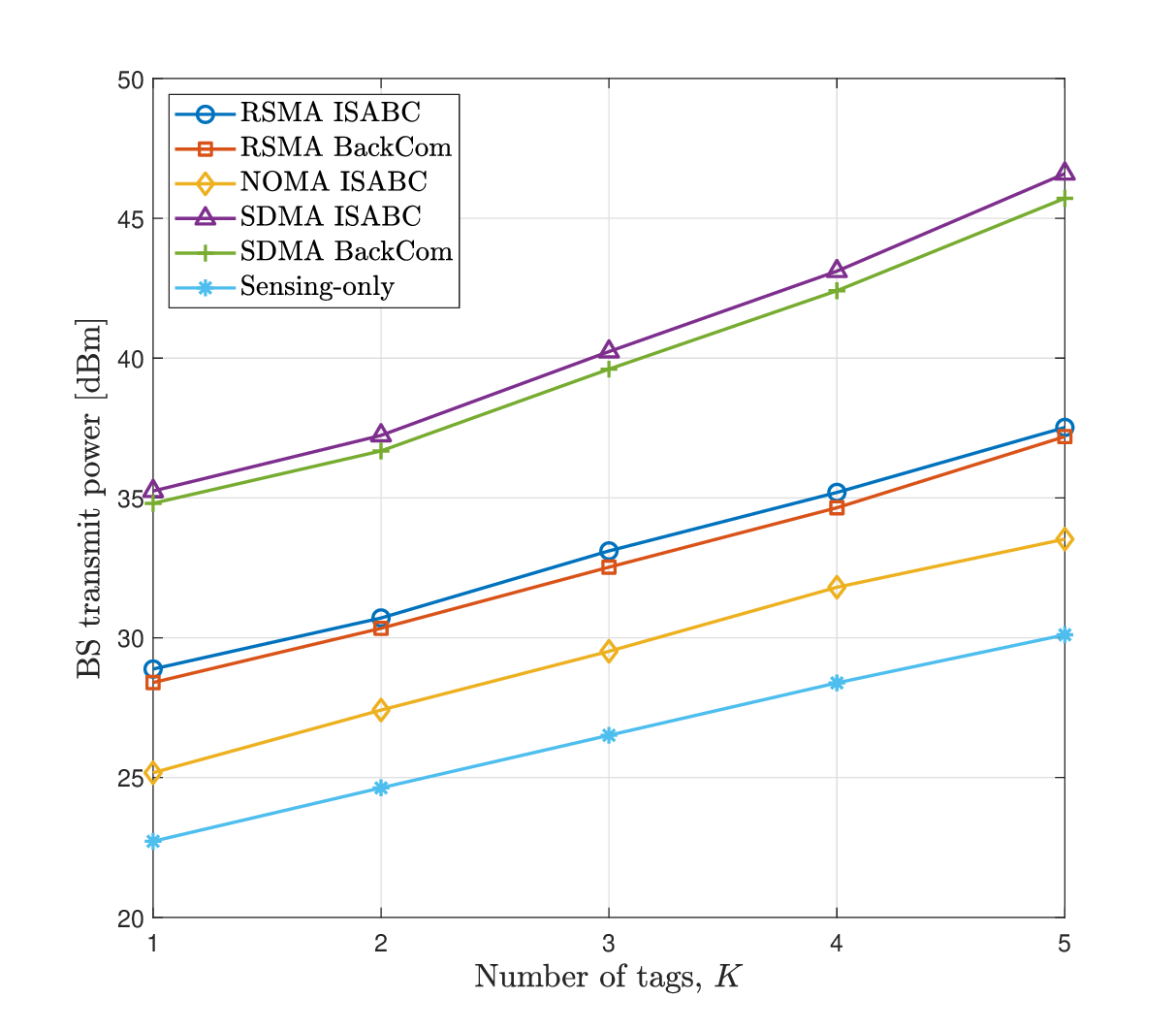}\vspace{-0mm}
    \caption{Transmit power versus the number of tags, $K$.}
	\label{fig_No_of_tag} \vspace{-0mm}
\end{figure}

Fig.~\ref{fig_No_of_tag} exhibits the intricate relationship between the number of tags and the BS transmit power requirements. From Fig.~\ref{fig_No_of_tag}, these two variables show an evident linear correlation across all schemes. Although the NOMA and sensing-only methods reduce the transmit power, they also have inferior communication performance (Fig.~\ref{fig_UserRate_M}). 

Conversely, both conventional \bc and conventional ISABC demand higher BS transmit power to satisfy communication (and/or sensing) rates and EH while mitigating multi-user multi-tag interference. In contrast, both the proposed and RSMA-assisted \bc utilize a moderate BS transmit power. Nonetheless,  compared to  RSMA-assisted \bc, the proposed ISABC approach requires a low power increase and promotes both communication and sensing by transforming passive tags into proactive participants. 

\subsection{Transmit Power Versus Number of Users}

\begin{figure}[!t]\vspace{-0mm}	
\centering
\includegraphics[width=0.45\textwidth]{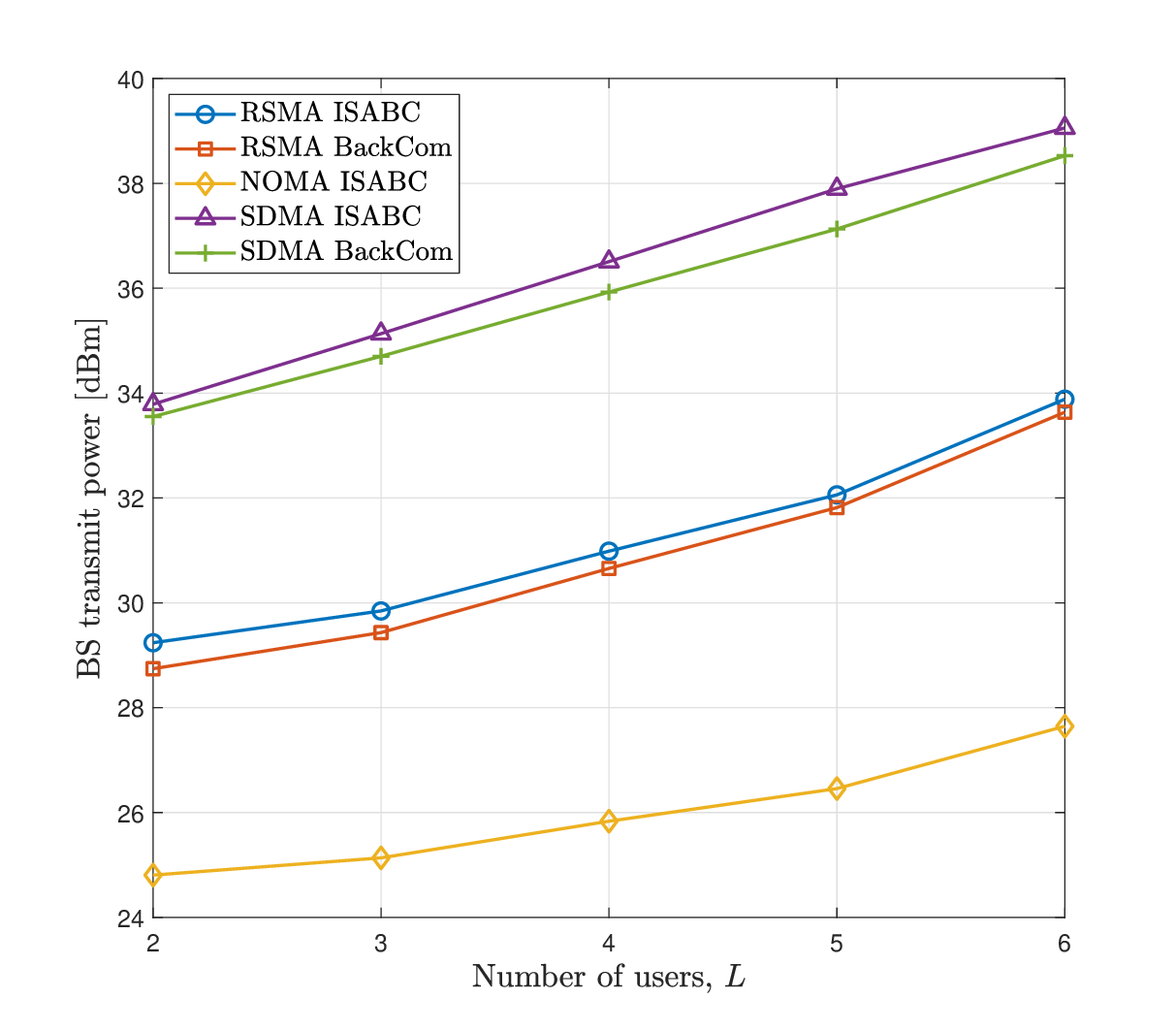}\vspace{-0mm}
    \caption{Transmit power versus the number of users, $L$.}
	\label{fig_No_of_users} \vspace{-0mm}
\end{figure}

{Fig.~{\ref{fig_No_of_users}} shows the relationship between the number of users and BS transmit power requirements. As the number of users increases, the BS requires more transmit power across all schemes. This is because serving more users demands additional energy to maintain adequate QoS, as power must be distributed across more communication links. Additionally, increased multi-user interference can degrade signal quality, requiring extra power to counteract this interference.} 

\subsection{Impacts of SIC and SI cancellation}

\begin{figure}[!t]\vspace{-0mm}	
\centering
\includegraphics[width=0.45\textwidth]{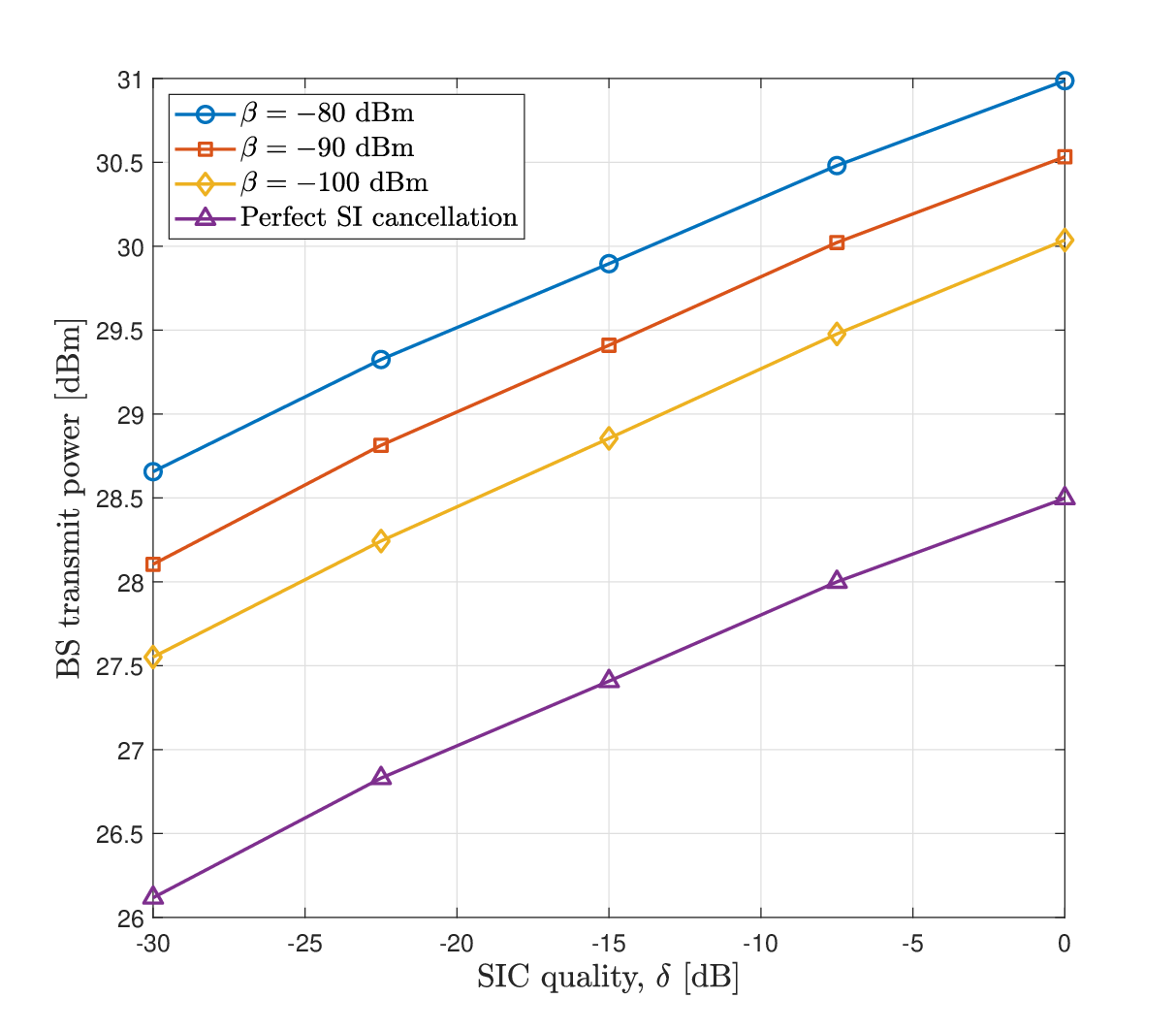}\vspace{-0mm}
    \caption{Transmit power versus SIC imperfection $\delta$ for various residual SI values.}
	\label{fig_SIC_quality} \vspace{-0mm}
\end{figure}

Impairments like SIC errors and SI cancellation quality degrade performance and reliability. SIC errors harm signal reception, while incomplete SI cancellation at FD nodes hinders incoming signal reception. Both errors have a significant impact.

Fig.~\ref{fig_SIC_quality} examines the effects of imperfect SIC and SI cancellation on Algorithm \ref{alg_AO}. The BS transmit power is plotted as a function of SIC quality, i.e., $\delta = \delta_c = \delta_p = \delta_s$, for various SI cancellation qualities, i.e., $\beta=\{-80, -90, -100\}\,\qty{}{\dB}$ and prefect SI cancellation, at the BS. 

According to Fig.~\ref{fig_SIC_quality}, severe SIC imperfection, i.e., $\delta \rightarrow 1\, (\qty{0}{\dB})$ requires the BS to transmit additional power to alleviate the adverse effects of SIC imperfection while maintaining tag EH requirements and communication/sensing rates. Conversely, when  $\delta \rightarrow 0$ or perfect SIC, less power is required to sustain these demands. Imperfect SI cancellation has bad consequences. Variable $\beta \in [0, 1]$ reflects the degree of imperfection at the BS. Further experiments offer insight into transmission power with varying residual SI values. Like SIC errors, the BS transmit power is sacrificed as $\beta$ increases.

\subsection{Impacts of Channel Conditions}

\begin{figure}[!t]\vspace{-0mm}	
\centering
\includegraphics[width=0.45\textwidth]{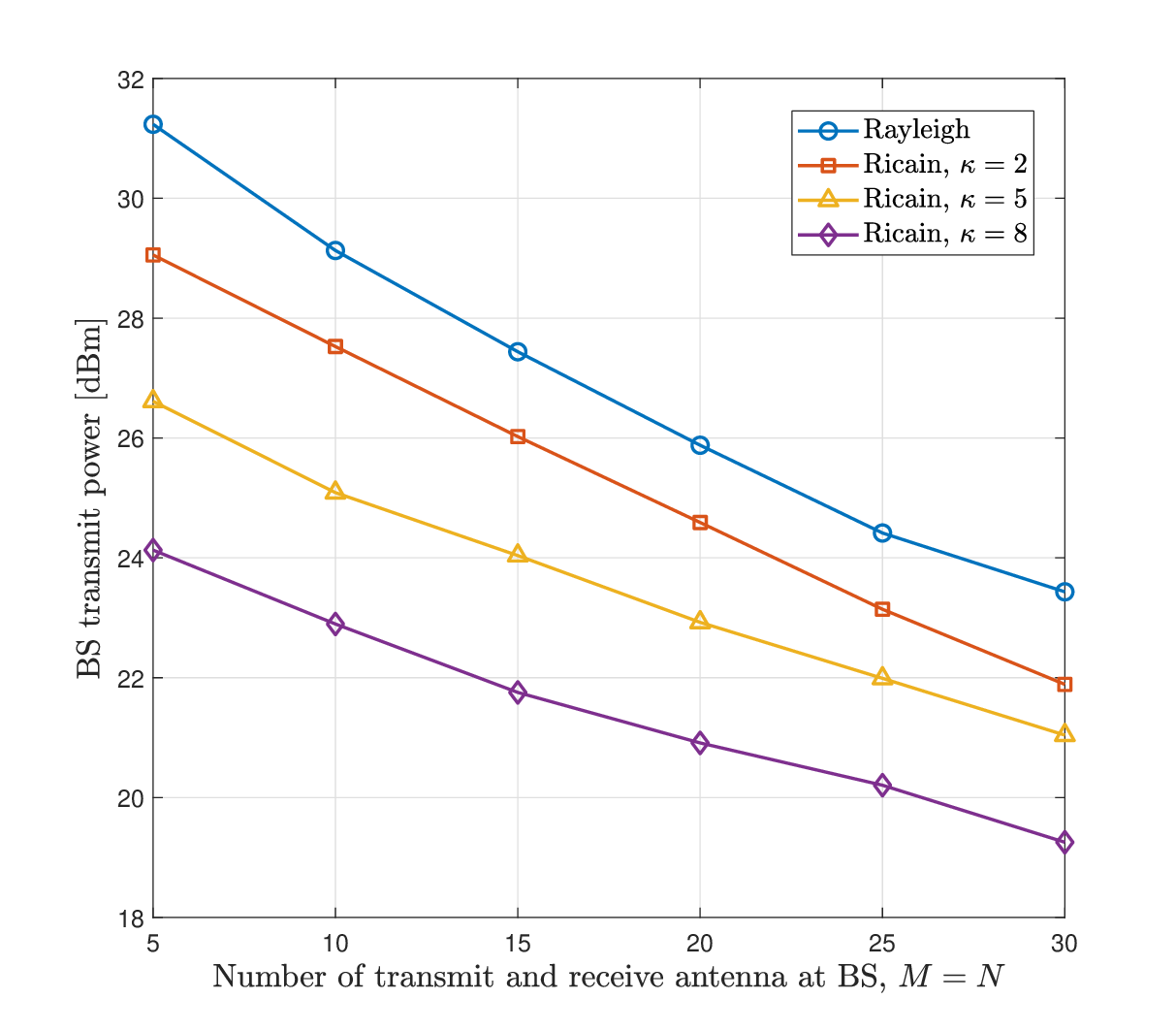}\vspace{-0mm}
    \caption{{Transmit power for different fading channels.}}
	\label{fig_ChannelCondition} \vspace{-0mm}
\end{figure}

{Fig.~{\ref{fig_ChannelCondition}} investigates the impact of different propagation environments. Rayleigh fading is suitable in dense urban and indoor environments with rich multipath scattering and no dominant LoS path. In contrast, Rician fading occurs in suburban, rural, or any setting with a clear LoS path, such as smart homes, farming, warehouses, and similar scenarios {\cite{Tse_Viswanath_2005}}. To this end, the pure communication channels are modeled as Rician fading with Rician factors $\kappa$}  
\begin{align}
    \mathbf{a}  & = \sqrt{\frac{\kappa}{\kappa+1}} \mathbf{a}^{\text{LoS}} + \sqrt{\frac{1}{\kappa+1}} \mathbf{a}^{\text{NLoS}}, 
\end{align}
{where $\mathbf{a} \in \{\mathbf{f}_0, \mathbf{f}_l, v_{l,k}, q_k\}$, $\mathbf{a}^{\text{LoS}}$ is the deterministic LoS components between the transmitter and receiver, and $\mathbf{a}^{\text{NLoS}}$ is the non-LoS (NLoS) components that follow the Rayleigh fading model. As per Fig.~{\ref{fig_ChannelCondition}}, high $\kappa$ results in lower transmit power at the BS. For instance, compared to Rayleigh channels, ${\kappa=\num{2}}$ and ${\kappa=\num{5}}$ reduce transmit power by {\qty{4.9}{\percent}} and {\qty{11.4}{\percent}}, respectively. Thus, the  LoS component improves the signal power utilization.}

\section{Conclusion} 
This paper presents a novel framework integrating RSMA and ISABC. In it, the FD BS  performs two functions:  first, it senses information from backscatter tags, and second, it serves multiple communication  users via RSMA. Conversely, the backscatter reader is responsible for detecting multi-tag backscatter data. An advanced AO-based algorithm is introduced to precisely regulate communication channels, both primary and backscatter, while concurrently reducing the transmit power of the BS. This innovation facilitates passive tags' opportunistic sensing and communication in future IoT networks. For comparative purposes,  the proposed system is benchmarked against  NOMA and SDMA, as well as sensing-only and communication-only approaches.

Rigorous comparisons show that  the RSMA-based system has significant advantages over NOMA and SDMA counterparts. Furthermore, the  impact analysis of SIC errors, the quality of SI cancellation, and different channel conditions are investigated. 

Future research can focus on ISABC channel estimation techniques, waveform/carrier signal design, signal processing at the BS/tag/reader, and developing key performance indicators to explore communication-sensing performance trade-offs. State parameter estimation methodologies can also be refined to enhance precision in dynamic environments. Additionally, evaluating system delay and EE will be vital for supporting future low-power IoT applications, where minimizing power consumption while maintaining reliable communication and sensing is critical for widespread deployment and sustainability in resource-constrained scenarios.

\appendices

\section{{Proof of Theorem 1: Algorithm Convergence}}\label{Theorem_1}
{The main problem is divided into three sub-problems, which optimize tag reflections coefficients  $\left(\qalpha := \{\alpha_k\}_{k\in \mathcal{K}}\right)$, received beamforming $\left( \qU :=\{\qu_k\}_{k\in \mathcal{K}} \right)$, and transmit beamforming $\left(\qP= \{\mathbf{w}_c, \{\mathbf{w}_l\}_{l\in \mathcal{ L}}, \mathbf{S}, \{C_l\}_{l\in \mathcal{ L}}\}\right)$, via solving problems {\eqref{eqn_optimal_u}}, {\eqref{Pw1_prob}}, and {\eqref{Palpha_prob}}, while keeping the other two blocks of variables fixed. Let us define $F(\qU, \qalpha, \qP)$ as a function of $\qU$, $\qalpha$, and $\qP$ for the objective value of {\eqref{P1_obj}}. First, in step $3$ of Algorithm {\ref{alg_AO}} with fixed variables $\qalpha^{(i)}$ and $\qP^{(i)}$, $\qU^{(i+1)}$ is the optimal solution that minimizes the value of the objective function. Accordingly, the following holds:}
\begin{equation}\label{43}
F(\qalpha^{(i)},\qU^{(i+1)},\qP^{(i)}) \leq F(\qalpha^{(i)},\qU^{(i)},\qP^{(i)}).
\end{equation}
{Next, in step $4$ of Algorithm {\ref{alg_AO}}, $\qP^{(i+1)}$ is the optimal transmit beamformers with given variables $\qalpha^{(i)}$ and $\qU^{(i+1)}$ to minimize $F$ via solving {\eqref{Pw_obj}}. Thus, it guarantees that}
\begin{equation}
F(\qalpha^{(i)},\qU^{(i+1)},\qP^{(i+1)}) \leq F(\qalpha^{(i)},\qU^{(i+1)},\qP^{(i)}).
\end{equation}
{Finally, in step $5$ of Algorithm {\ref{alg_AO}} with the given $\qP^{(i+1)}$ and  $\qU^{(i+1)}$, problem {\eqref{Pw1_prob}} is solved to obtain an optimal solution for $\qalpha^{(i)}$, which yields:}
\begin{equation}\label{45}
F(\qalpha^{(i+1)},\qU^{(i+1)},\qP^{(i+1)}) \leq F(\qalpha^{(i)},\qU^{(i+1)},\qP^{(i+1)}).
\end{equation}
{According to {\eqref{43}}--{\eqref{45}}, it follows  that}
\begin{equation} 
\!\!\!\!F(\qalpha^{(i+1)},\qU^{(i+1)},\qP^{(i+1)}) \leq F(\qalpha^{(i+1)},\qU^{(i+1)},\qP^{(i+1)}).
\end{equation}
{The objective values of Algorithm {\ref{alg_AO}} monotonically decrease with each iteration, always remaining non-negative. This consistency, combined with the design choice where each iteration starts from the previous one's end, ensures the algorithm's convergence. Thus, the proof is completed.}

\bibliographystyle{IEEEtran}
\bibliography{IEEEabrv,ref}

\begin{IEEEbiography}[{\includegraphics[width=1in,height=1.2in,clip,keepaspectratio]{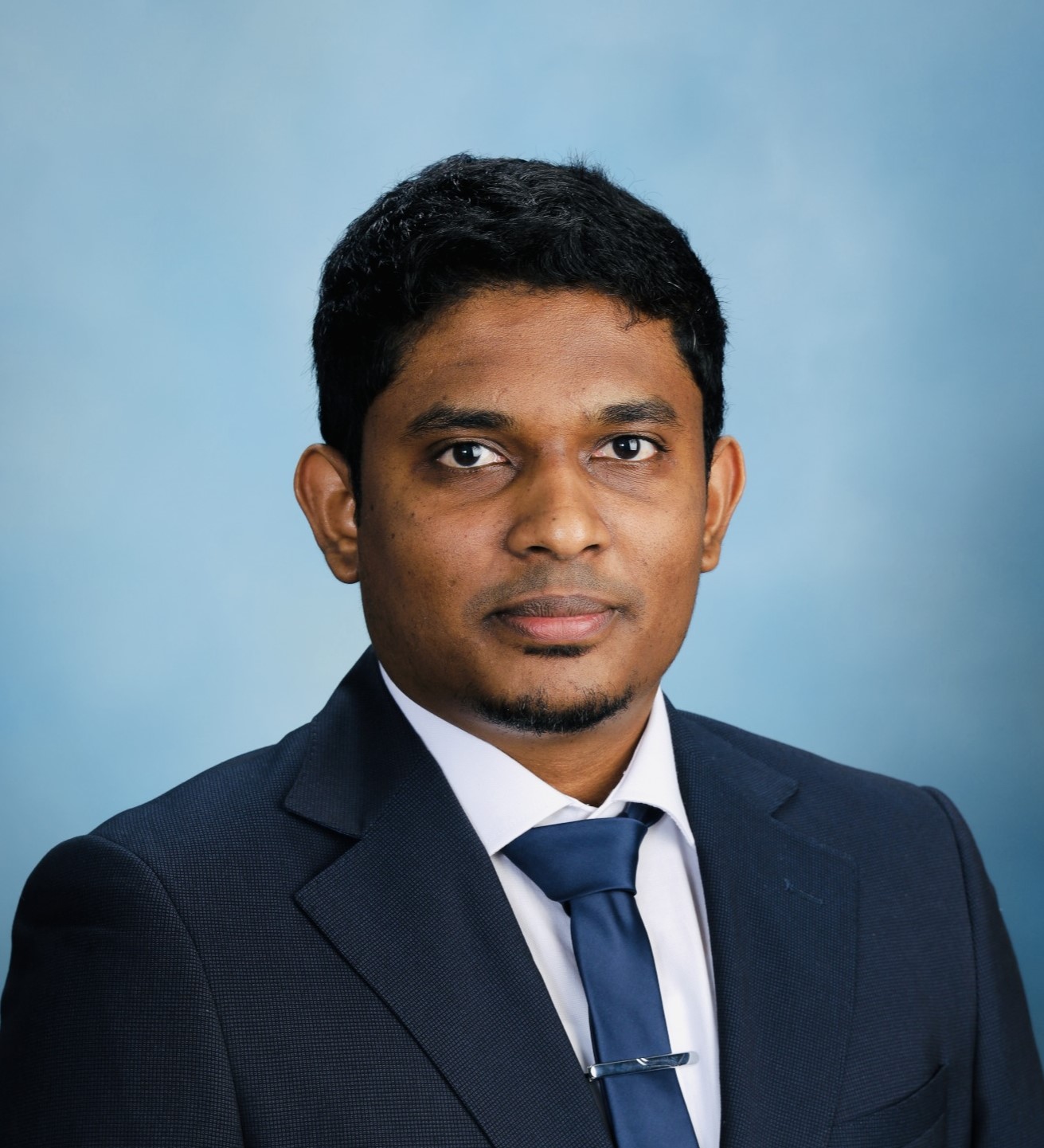}}]{Diluka  Galappaththige} \;(S'17--M'22) is a postdoctoral research fellow at the Department of Electrical and Computer Engineering, University of Alberta, Canada. He received his B.Sc. degree (First-class honor) in Electrical and Electronic Engineering from the Department of Electrical and Electronic Engineering, University of Peradeniya, Sri Lanka, in 2017, and his Ph.D. in Electrical and Computer Engineering from the School of Electrical, Computer, and Biomedical Engineering, Southern Illinois University, Carbondale, IL, USA, in 2021. 

Dr. Galappaththige has been awarded a post-doctoral fellowship from NSERC for the academic year 2024-2026. His current research interests include but are not limited to, the design, modeling, and analysis of massive multiple-input multiple-output (mMIMO) communication (i.e., including co-located mMIMO and cell-free/distributed mMIMO), full-duplex communication (FD), backscatter communication (BackCom), reconfigurable intelligent surfaces (RISs), integrated sensing and communication (ISAC), near-field ISAC, wireless power transfer, and emerging technologies for enabling fifth-generation (5G) and beyond wireless networks. He was a recipient of the Exemplary Reviewer Award for IEEE Wireless  Communications Letters (TWCL) in 2020, and IEEE  Communications Letters (TCL) in 2021. He has actively served as a reviewer for a variety of IEEE journals and conferences, including IEEE Transactions on Communications, IEEE Internet of Things Journal, IEEE Transactions on Vehicular Technology, IEEE Transactions on Green Communications and Networking, IEEE Access, IEEE Network, IEEE Open Journal of Communications Society, TCL, TWCL, IEEE International Conference on Communications, and IEEE Global Communications Conference.
\end{IEEEbiography}

\begin{IEEEbiography}[{\includegraphics[width=1in,height=1.25in,clip,keepaspectratio]{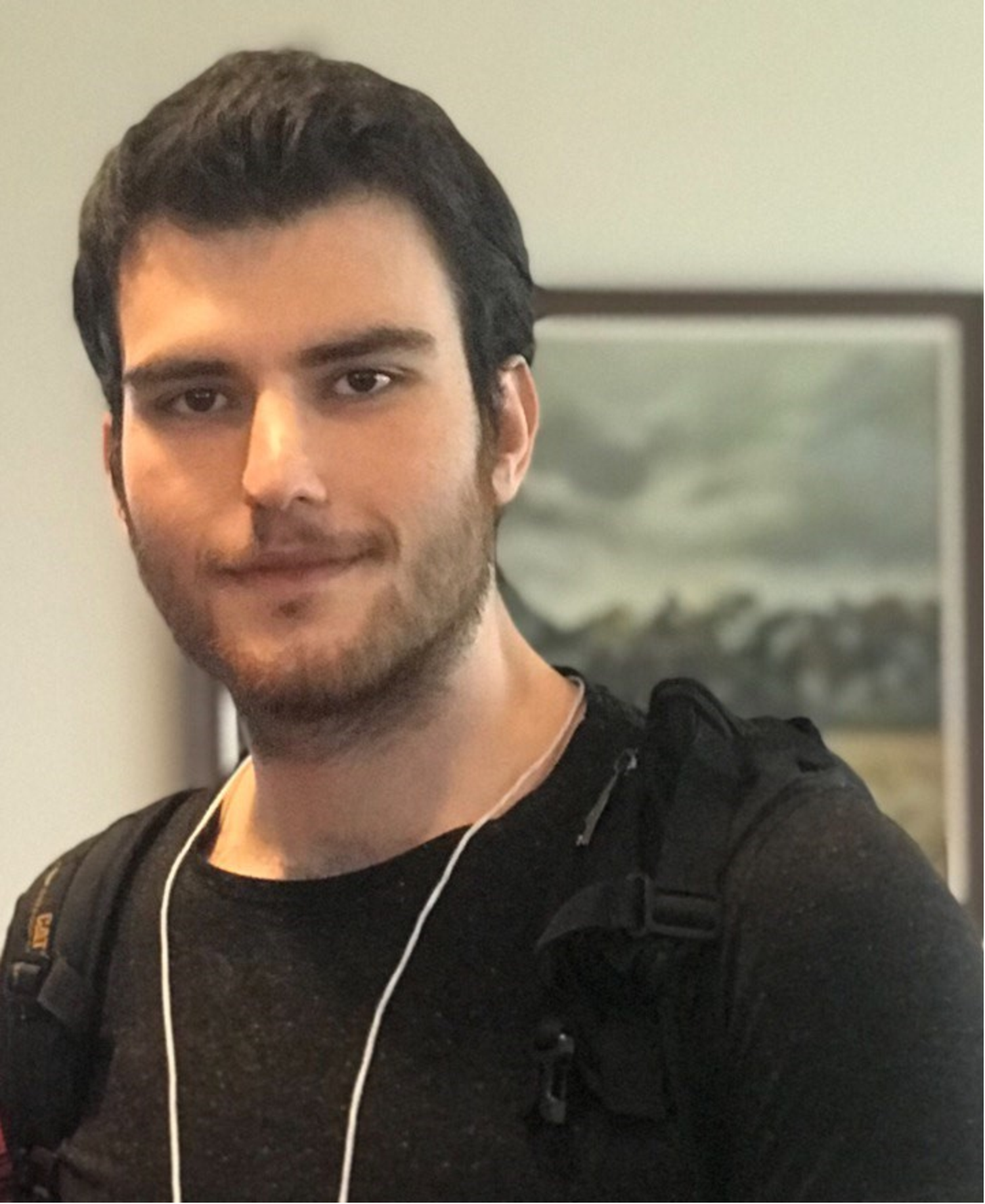}}]
{Shayan Zargari}\; received the B.Sc. degree in Electronic Engineering from Azad University, South Tehran Branch, Tehran, Iran, in 2018, and the M.Sc. degree in Telecommunication Engineering from the Iran University of Science and Technology, Tehran, Iran, in 2020. He is currently pursuing the Ph.D. degree in communications and signal processing at the University of Alberta, Edmonton, AB, Canada. 

From 2019 to 2020, he was a visiting researcher at the Electronics Research Institute, Sharif University of Technology, Tehran, Iran, and from 2020 to 2021, he was a visiting researcher at the Department of Electrical and Computer Engineering, Tarbiat Modares University, Tehran, Iran. His research interests include optimization theory, integrated sensing and communication (ISAC), backscatter communication (BackCom), intelligent reflecting surfaces (IRS), unmanned aerial vehicle (UAV) communications, resource allocation in fifth-generation (5G)/six-generation (6G) wireless communication, and green communication. He serves as a reviewer for several IEEE journals, including IEEE JOURNAL ON SELECTED AREAS IN COMMUNICATIONS, IEEE TRANSACTIONS ON COMMUNICATIONS, and IEEE TRANSACTIONS ON WIRELESS COMMUNICATIONS.
\end{IEEEbiography}

\begin{IEEEbiography}[{\includegraphics[width=1in,height=1.25in,clip,keepaspectratio]{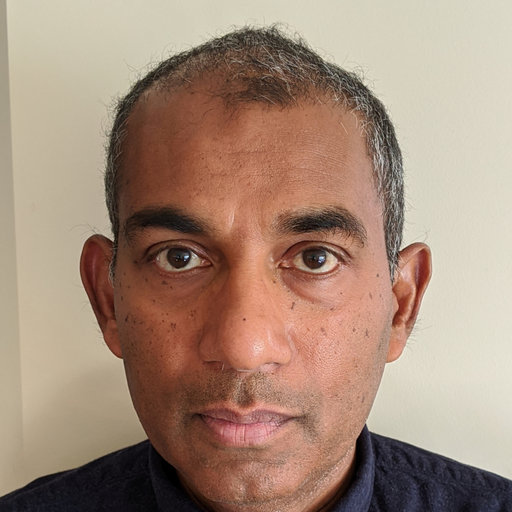}}]{Chintha Tellambura} \;(Fellow, IEEE) received the B.Sc. degree (First class) in electrical and electronic engineering from the University of Moratuwa, Sri Lanka, the  M.Sc. degree in electrical engineering from King’s College, University of London, and the  Ph.D. degree in electrical engineering from the University of Victoria, Canada. He was with Monash University, Australia, from 1997 to 2002. 
Dr. Tellambura is a  Professor in the Department of Electrical and Computer Engineering at the University of Alberta. He has authored or co-authored over 650 journals and conference papers, demonstrating his expertise in the field. According to Google Scholar, his exceptional scholarly contributions have earned him an impressive H-index of 84. Dr. Tellambura has made significant contributions to various areas of research, including future wireless networks, machine learning for wireless networks, and signal processing. 

Recognizing his outstanding accomplishments, he was elected as an IEEE Fellow in 2011 for his noteworthy contributions to physical layer wireless communication theory. In 2017, he was further honored as a fellow of the Canadian Academy of Engineering, a testament to his exceptional achievements.  His dedication and expertise have been acknowledged through prestigious awards, including the Best Paper Awards in the Communication Theory Symposium in 2012, the IEEE International Conference on Communications (ICC) held in Canada in 2017, and another ICC in France. Moreover, Dr. Tellambura has been honored with the esteemed McCalla Professorship and the Killam Annual Professorship by the University of Alberta, further underscoring his significant impact on academia. Dr. Tellambura has also played a vital role in editorial responsibilities within the IEEE community. He served as an Editor for the IEEE Transactions on Communications from 1999 to 2011 and for the IEEE Transactions on Wireless Communications from 2001 to 2007. In the latter role, he was Area Editor of Wireless Communications Systems and Theory from 2007 to 2012, contributing to advancing the field through his editorial expertise.
\end{IEEEbiography}

\begin{IEEEbiography}[{\includegraphics[width=1in,height=1.25in,clip,keepaspectratio]{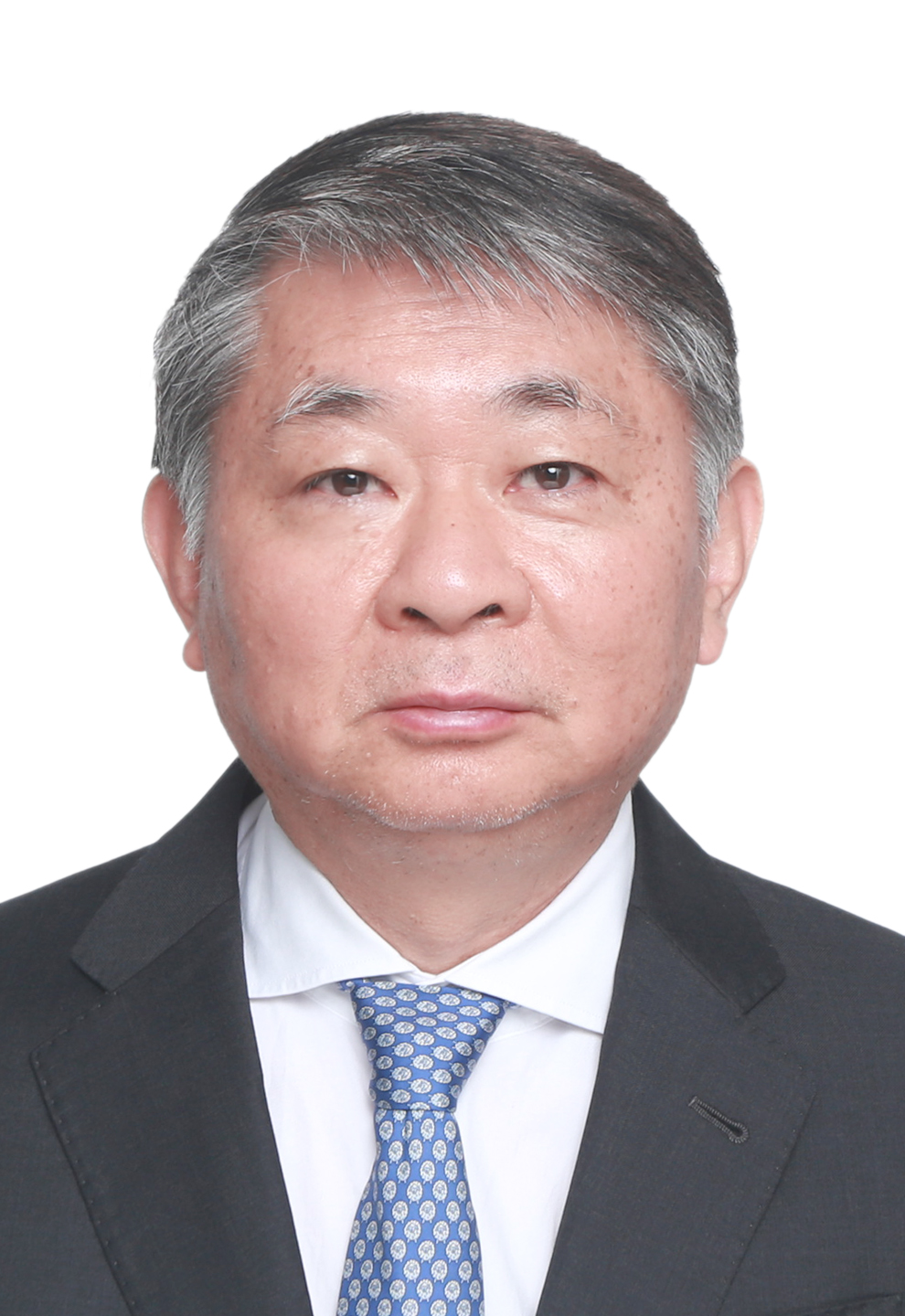}}]{Geoffrey Ye Li} \;(Fellow, IEEE) is currently a Chair Professor at Imperial College London, UK. Before joining Imperial in 2020, he was a Professor at Georgia Institute of Technology, USA, for 20 years and a Principal Technical Staff Member with AT\&T Labs – Research (previous Bell Labs) in New Jersey, USA, for five years. He made fundamental contributions to orthogonal frequency division multiplexing (OFDM) for wireless communications, established a framework on resource cooperation in wireless networks, and introduced deep learning to communications. In these areas, he has published around 700 journal and conference papers in addition to over 40 granted patents. His publications have been cited over 75,000 times with an H-index of 127. He has been listed as a Highly Cited Researcher by Clarivate/Web of Science almost every year.

Dr. Geoffrey Ye Li was elected to Fellow of the Royal Academic of Engineering (FREng), IEEE Fellow, and IET Fellow for his contributions to signal processing for wireless communications. He received 2024 IEEE Eric E. Sumner Award, 2019 IEEE ComSoc Edwin Howard Armstrong Achievement Award, and several other awards from IEEE Signal Processing, Vehicular Technology, and Communications Societies.
\end{IEEEbiography}

\end{document}